\documentclass[a4paper,twocolumn,11pt,accepted=2025-11-04]{quantumarticle}
\pdfoutput=1
\usepackage[utf8]{inputenc}
\usepackage[english]{babel}
\usepackage[T1]{fontenc}
\usepackage{amsmath}
\usepackage{hyperref}

\usepackage{tikz}
\usepackage{lipsum}
\usepackage[margin=1in]{geometry}
\usepackage{amsthm}
\usepackage{amsmath} 
\usepackage{braket}
\usepackage[toc,page]{appendix}
\usepackage{xcolor}
\usepackage{physics}
\usepackage{hyperref}     
\usepackage{amsfonts,amsmath,amssymb}     
\usepackage{float}
\usepackage{subcaption}

\usepackage{graphicx}
\usepackage{wrapfig}
\usepackage{color,framed}
\usepackage{bbm}
\usepackage{bm}
\usepackage{blindtext}
\usepackage{titlesec}

\usepackage[capitalize]{cleveref}

\usepackage{xcolor}
\hypersetup{
    colorlinks=true,
    linkcolor=blue,
    citecolor=blue, 
    urlcolor=blue   
} 

\usepackage{etoolbox}

\newtheorem{remark}{Remark}
\newtheorem{theorem}{Theorem}
\newtheorem{informal_theorem}{Informal Theorem}

\newtheorem{lemma}{Lemma}

\newcommand{\E}{\mathbb{E}}

\usepackage{xcolor}

\definecolor{mydarkgreen}{rgb}{0.0, 0.5, 0.0}

\begin{document}

\title{Direct Analysis of Zero-Noise Extrapolation: Polynomial Methods, Error Bounds, and Simultaneous Physical–Algorithmic Error Mitigation}

\author{Pegah Mohammadipour}
\email{pegahmp@psu.edu}
\affiliation{Department of Mathematics, The Pennsylvania State University,  University Park, Pennsylvania 16802, USA}
\orcid{0009-0008-0227-1310}

\author{Xiantao Li}
\email{xiantao.li@psu.edu}
\affiliation{Department of Mathematics, The Pennsylvania State University,  University Park, Pennsylvania 16802, USA}
\orcid{0000-0002-9760-7292}

\begin{abstract}
Zero-noise extrapolation (ZNE) is a widely used quantum error mitigation technique that artificially amplifies circuit noise and then extrapolates the results to the noise-free circuit. A common ZNE approach is Richardson extrapolation, which relies on polynomial interpolation. Despite its simplicity, efficient implementations of Richardson extrapolation face several challenges, including approximation errors from the non-polynomial behavior of noise channels, overfitting due to polynomial interpolation, and exponentially amplified measurement noise. This paper provides a comprehensive analysis of these challenges, presenting bias and variance bounds that quantify approximation errors. Additionally, for any precision $\varepsilon$, our results offer an estimate of the necessary sample complexity. We further extend the analysis to polynomial least squares-based extrapolation, which mitigates measurement noise and avoids overfitting. Finally, we propose a strategy for simultaneously mitigating circuit and algorithmic errors in the Trotter-Suzuki algorithm by jointly scaling the time step size and the noise level. This strategy provides a practical tool to enhance the reliability of near-term quantum computations. We support our theoretical findings with numerical experiments.
\end{abstract}

\maketitle

\renewcommand{\contentsname}{\centering Contents}

\section{Introduction} 

Quantum computing promises substantial speedups for a wide range of scientific computing challenges \cite{preskill2018nisq}. However, the noise inherent in near-term quantum devices presents a significant challenge, requiring innovative strategies to achieve reliable results. Quantum error mitigation (QEM) \cite{li2017efficient, KTemme_SBravyi_JGambetta} emerges as a promising solution, substantially reducing the impact of noise without the large qubit overhead demanded by full quantum error correction \cite{surface_codes}. By bridging the gap between theoretical speedups and practical performance, robust QEM schemes play a critical role in unlocking the true potential of quantum algorithms.

To model noise in quantum circuits, one evaluates the effects of the environment on the density operator \(\rho\). In the Markovian regime, the system’s dynamics under noisy evolution are described by the Lindblad equation \cite{lindblad1976generators,gorini1976completely}:  
\begin{equation}\label{eq:lindblad_eq}
    \frac{d}{dt} \rho 
    = -i[H,\rho]+\lambda \sum_{j} \left(V_j\rho V^{\dag}_j-\frac{1}{2}\{ V_j^{\dag}V_j,\rho\}\right),
\end{equation}
This continuous-time description is motivated by the theory of open quantum systems \cite{breuer2002theory}.  
For many circuit noise models, such as depolarizing channels, the jump operators \( V_j \) can be explicitly identified \cite{EvanB_MiladM}. The parameter \(\lambda\) represents the noise strength or coupling constant, determining the rate at which the system interacts with its environment.

One simple strategy to mitigate noise is zero-noise extrapolation (ZNE) \cite{li2017efficient,KTemme_SBravyi_JGambetta,endo2018practical}. ZNE has been successfully incorporated into various quantum algorithms, including Hamiltonian simulation \cite{EndoMitigAlgoError}, variational quantum algorithms \cite{kurita2022synergetic}, and quantum linear system solvers \cite{vazquez2022enhancing,carrera2022extrapolation}.

The idea of ZNE can be illustrated by expressing the density matrix as \(\rho(t,\lambda)\), where \(t\) denotes the evolution time and \(\lambda\) the noise strength. Although \(\rho\) is \emph{unknown}, one can amplify the effective noise level associated with gate operations—e.g., by gate folding— so that the circuit prepares \(\rho(t, x\lambda_0)\) for some scaling factor \(1 \le x \le B\), where \(\lambda_0\) is the initial noise level of the hardware \cite{KTemme_SBravyi_JGambetta}. Denoting the expectation value of an observable \(A\) measured on the resulting state \(\rho(t, x\lambda_0)\) by
\begin{equation}\label{eq: E(x)}
E(x) = \operatorname{Tr}\!\left(A\rho(t, x\lambda_0)\right), \qquad x \in [1,B],
\end{equation}
the goal is to estimate the ideal, noise-free value \(E(0)\) from the measured data \(E(x_0), E(x_1), \dots, E(x_n)\), where each \(x_j\) corresponds to a different noise amplification.

These values are then extrapolated to the zero-noise limit using polynomial interpolation or regression.
Such an extrapolation procedure, while straightforward, comes with several caveats. First, except for specific noise channels, the expectation value is generally not a polynomial function, leading to inevitable approximation errors. Second, expectation values are obtained through measurements, which are inherently subject to shot noise. In such cases, direct polynomial interpolation is prone to overfitting, a well-known issue in statistical learning \cite{hastie2009elements}. Third, the extrapolation coefficients \(\{\gamma_k\}_{k=0}^n\) tend to grow rapidly with \(n\) \cite{gautschi1987lower,li2006lower}, exponentially amplifying statistical errors as the number of interpolation points increases.

Note that overfitting and noise amplification are closely related in practice, 
yet distinct: overfitting reflects a bias due to unstable polynomial interpolation when the polynomial degree is large, 
while noise amplification corresponds to the variance growth induced by large extrapolation coefficients. 
We analyze these effects separately and then use a concentration bound to relate them.

Noise scaling for ZNE can be achieved through unitary folding, and parameterized noise scaling \cite{digital_zero_noise_extrap}. Moreover, in Richardson extrapolation, the number of measurements, the choice of interpolation nodes, and their distribution all influence the approximation error \cite{MichaelKrebsbachOptimizingRichardson}.

While many empirical studies have evaluated the performance of ZNE \cite{digital_zero_noise_extrap,MichaelKrebsbachOptimizingRichardson}, a comprehensive theoretical framework, to the best of our knowledge, is still lacking. 

\subsection{Main Results and Assumptions}

With the assumption of Markovian noise \eqref{eq:lindblad_eq}, this paper aims to address several important questions toward a thorough understanding of ZNE. 

\noindent{\bf Problem 1.} How many interpolation points (circuits) and samples per circuit are needed to achieve an approximation of precision \(\varepsilon\)? 

\medskip

We address this question through direct analysis of ZNE, with the results summarized in the following theorem:

\begin{informal_theorem} \label{thm:info_1}[Informal version of Theorem \ref{thm:Hoeffding_Richardson}]
  Suppose the expectation value of the observable depends smoothly on the noise parameter (i.e., it is Gevrey-class regular), corresponding to circuits of modest depth or sufficiently low noise levels.
  Then the error of Zero-Noise extrapolation is less than \(\varepsilon\) with high probability \(1 - \delta\), provided the number of interpolation points is at least \(\Omega \!\left( \log \tfrac{1}{\varepsilon} \right)\), and the number of samples satisfies:
  \begin{itemize}
      \item \(N_S=\Omega\left( \varepsilon^{-(2+ 2e(\kappa+1)}\log \frac{2}{\delta}\right)\) for equidistant interpolation points,
      \item \(N_S=\Omega\left(  \varepsilon^{-(2 + 4 \log \kappa)}\log \frac{2}{\delta}\right)\) for Chebyshev interpolation points, 
  \end{itemize}
    where \(\kappa = \frac{\sqrt{B} + 1}{\sqrt{B} - 1}\), with $B$ given in \cref{eq: E(x)}.
\end{informal_theorem}
We further validate and explain the implications of the smoothness condition in the context of quantum error mitigation following Theorem~\ref{thm:Hoeffding_Richardson}.

\bigskip 

Another issue, from a statistical learning perspective, is that polynomial interpolation tends to overfit \cite{hastie2009elements}, even when a moderate number of data points is used. To address this, we adopt a least-squares framework, ensuring that the polynomial degree remains significantly lower than the number of data points. This leads to a natural question:

\medskip

\noindent{\bf Problem 2.} How does the choice of the approximation method, such as polynomial interpolation or least squares, affect the bias and variance?

\medskip

To address this, we analyze the approximation and statistical error using properties of Chebyshev polynomials \cite{Trefethen_approx}. More specifically, for a polynomial \( p_m \) of degree \( m < n \), obtained through the least-squares method, we derive the following bound related to the variance:

\begin{informal_theorem} (Informal version of Theorem \ref{thm:LSA_gamma_sum})  
The least-squares estimator using Chebyshev nodes satisfies the bound  
\[
        \sum_{k=0}^{n} |\gamma_k| = O\left(\kappa^{2m}\right).
\]
To ensure that the bias remains within \(\varepsilon\), Theorem \ref{thm:bias_bound_lsa} establishes that the dependence of \(m\) on \(\varepsilon\) is logarithmic.  
\end{informal_theorem}

Thus, by restricting the degree of least-squares polynomials to lower degrees, we prevent the rapid growth of statistical error. Consequently, we can determine the overall sampling complexity in a manner similar to Informal Theorem 1.

\medskip
Apart from errors due to circuit noise, the results of many quantum algorithms are subject to algorithmic errors that arise from approximations inherent to them. One important example is Trotter decomposition for quantum simulation, where finite step sizes introduce discretization errors that accumulate over long-time evolutions \cite{Childs_Trotter_Theory_2021}. Moreover, algorithmic and physical errors often interact, as the deeper circuits required to reduce algorithmic errors increase susceptibility to physical noise \cite{EndoMitigAlgoError}.  This observation motivates the following question:

\medskip
\noindent{\bf Problem 3.} Can physical and algorithmic errors be jointly mitigated?

\medskip
The algorithmic error resulting from the Trotter-Suzuki approximation can be mitigated using ZNE \cite{ImprovedAccuracyforTrotterSim,watson2024exponentially}. In fact, our analysis parallels that of \cite{watson2024exponentially}, with the key distinction that in our case the physical error is inherently nonnegative and the scaling factor cannot drop below 1. While \cite{EndoMitigAlgoError} implemented both mitigation strategies, they treated them as separate procedures.
 In this work, we extend ZNE to concurrently address both algorithmic and physical errors in Hamiltonian simulation based on the second-order Trotter-Suzuki approximation. concurrently

\begin{informal_theorem}[Informal version of Theorem \ref{thm:joint_Hoeffding}]
  Suppose the expectation value of observable \(A\) affected by both circuit and Trotter errors, admits a Gevrey-class smooth dependence (corresponding to circuits of modest depth or sufficiently low noise), and the extrapolation is performed using Chebyshev interpolation nodes in the interval $[1,B]$. Then the error of Richardson extrapolation, accounting for both errors, is less than \(\varepsilon\) with high probability \(1 - \delta\) if the number of interpolation points is at least \( \Omega\left(\log \frac1\varepsilon\right)\), and the number of samples is at least  \(\Omega\left(\alpha^2  \varepsilon^{-(2+4 \log \kappa)}\log\frac{2}{\delta}\right)\) (\(\alpha\) is a norm bound on the observable, \(\|A\|_{\infty} \leq \alpha\), and \(\kappa\) is as in Informal Theorem \ref{thm:info_1}).
 
\end{informal_theorem}

Previously, it has been shown that the number of samples required to mitigate depolarizing or non-unital noise, scales exponentially with the number of qubits and circuit depth \cite{Universal_sampling_lower_bounds_QEM,takagi2022fundamental,Quek_tighter_bounds_2024}. These bounds hold regardless of the specific QEM strategy. Our analysis, however, provides direct estimates of the sampling complexity for particular QEM methods. 

We emphasize that our goal is to establish bounds on both approximation and statistical errors rather than to overcome the exponential scaling with circuit depth. The fundamental limitations of QEM, as highlighted in previous works, remain in effect. Nevertheless, by directly linking the properties of noise dynamics and algorithmic error to sampling complexity, we hope our analysis will help identify noise models and algorithms that are more amenable to QEM procedures.

\subsection{Notations and Terminologies.}
In our analysis, we will use the following standard notations. 
\begin{itemize}
    \item The Hilbert space for an \(n\)-qubit system is 
    \[
    \mathcal{H} = \left(\mathbb{C}^2\right)^{\otimes n} 
    \]
    \item The observable \(A\) is bounded, \(\|A\|_{\infty} \leq \alpha\) (the Schatten infinity norm).
    \item The final evolution time on the circuit is \(T\), resulting in the state \(\rho(T, \lambda)\).
    \item The Lindblad operator \(\mathcal{L}\) is bounded, \(\|\mathcal{L}\|_1 \leq l\) (the induced Schatten 1-norm).
\end{itemize}

\paragraph{Organization.} The rest of the paper is organized as follows. In Section \ref{sec:Richardson_extrapplation}, we establish a bias bound that quantifies the approximation error in Richardson extrapolation, followed by a variance bound that provides an estimate of the sampling complexity. To mitigate overfitting, Section \ref{sec:least squares} explores the least squares method, particularly for a large number of nodes. In Section \ref{sec:Trotter_error}, we address the simultaneous mitigation of both circuit and algorithmic errors, using the Trotter-Suzuki decomposition as a practical example. Section \ref{sec:Numerics} presents numerical experiments that validate our theoretical findings. Finally, we summarize our results and discuss their implications in \cref{sec:summary_discussions}. Additional details, including derivative bounds and the proofs of lemmas and theorems, are provided in \cref{appendix}.

\section{Richardson's Extrapolation} \label{sec:Richardson_extrapplation}

Overall, the goal of quantum error mitigation is to estimate the expectation value of an observable \(A\) with respect to the noisy evolved state \(\rho(T,\lambda)\) at time \(T\), where the circuit noise is parameterized by \(\lambda = x\lambda_0\). Here \(\lambda_0\) represents the physical noise level of the unmodified circuit, while \(x \ge 1\) is an amplification factor—typically implemented through gate folding or related techniques \cite{digital_zero_noise_extrap,larose2022mitiq}. The ideal outcome corresponds to the zero-noise limit,
\[
\lim_{x\to 0^+} E(x).
\]
To formalize our analysis, we recast this into a function-approximation problem by denoting
\begin{equation}\label{eq:f-definition}
    f(x) := E(x),
\end{equation}
and seek to estimate \(f(0)\) from evaluations of \(f(x)\) at discrete amplified noise levels.

Richardson’s extrapolation \cite{Richardson1911,Richardson1927,Extrapolation_book} is a classical method in numerical analysis for accelerating the convergence of an approximation \( f^* = \lim_{x\to 0} f(x) \) using finitely many samples of \(f(x)\) at points \(\{x_j\}\). In the context of QEM \cite{KTemme_SBravyi_JGambetta}, this process is illustrated in Fig.~\ref{fig:RE}. Given the values of \(f(x)\) at \(1 = x_0 < x_1 < \dots < x_n \in [1,B]\), the method constructs a polynomial \(p_n(x)\) of degree at most \(n\) that interpolates these points and extrapolates its value to \(x=0\) as an estimate of \(f^* = f(0)\).\footnote{Interpolation and extrapolation both approximate unknown values from known data: interpolation estimates within the data range, whereas extrapolation extends beyond it. Since both use the same approximating polynomial, we use the terms interchangeably.}

We restrict the scaling parameter to \(x \in [1,B]\), where \(B\) denotes the maximal feasible amplification limited by circuit depth. For clarity, we first present the general theoretical framework and later discuss its implications for quantum error mitigation.

\begin{figure}[H]
    \centering
    \includegraphics[scale=0.4]{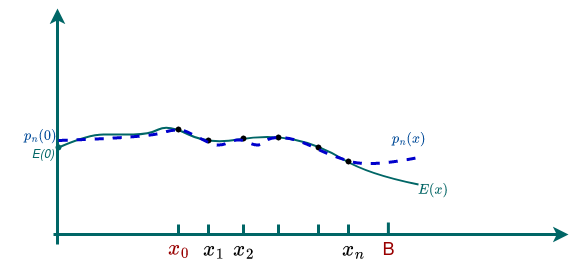}
    \caption{A simple illustration of Richardson’s extrapolation for error mitigation.
        A polynomial \( p_n(x) \) is constructed by interpolating \( f(x) \) at the points \( x_0, x_1, \dots, x_n \) and then evaluated at \( x = 0 \), corresponding to the zero-noise limit.
}
    \label{fig:RE}
\end{figure}

In its Lagrange form, the polynomial interpolant can be written explicitly as
\begin{equation} \label{def:extrap-pn}
p_n(x) =  \sum_{j=0}^{n} {f}(x_j) L_j(x),
\end{equation}
where \( L_j(x) \) are the standard Lagrange basis polynomials. The extrapolated approximation at zero is then given by
\[
f^* \approx p_n(0),
\]
which, from \eqref{def:extrap-pn}, simplifies to
\begin{equation} \label{extrap-pn0}
p_n(0) =  \sum_{j=0}^{n} {f}(x_j) \gamma_j, 
\end{equation}
where
\begin{equation} \label{gamma_j}
    \gamma_j  := L_j(0) = \prod_{\substack{k=0 \\ k \neq j}}^n \frac{x_k}{x_k - x_j}.
\end{equation}

The above definition is equivalent to the condition \(\sum_{k=0}^{n} \gamma_k x_k^{r} = \delta_{r,0}, \quad r = 0,1,\dots,n.\)
In practice, the function values (e.g., in QEM and many other applications) must be estimated from sampled data. To emphasize this, we rewrite \eqref{extrap-pn0} as
\begin{equation} \label{extrap-pn0-samp}
\hat{p}_n(0) :=  \sum_{j=0}^{n} \hat{f}(x_j) \gamma_j, 
\end{equation}
where \( \hat{f}(x_j) \) is an unbiased statistical estimator of \( {f}(x_j) \), and the corresponding statistical error is given by \( \delta_j = \hat{f}(x_j) - f(x_j) \). That is,
\[
\hat{f}(x_j) = \frac{1}{N_S} \sum_{i=1}^{N_S} f^{(i)}(x_j),
\]
where \( f^{(i)}(x_j) \) represents the \( i \)-th sample (e.g., from the \( i \)-th measurement in QEM), and \( N_S \) is the number of samples (e.g., shots in circuit execution).

We now encounter a familiar scenario from statistical analysis: selecting the appropriate \( n \) and estimator \( \hat{f}(x_j) \) such that both the bias and variance remain within a specified tolerance, ensuring that the error \( \abs{f^* - \hat{p}_n(0)} \) is small with high probability. In what follows, we analyze these error contributions.

\medskip 

We follow the standard theory of polynomial interpolation \cite{Atkinson}, which provides the following general form for the interpolation error. For every function \( f \in C^{n+1}(\mathbb{R}) \), there exists \( \xi \in [x_0, x_n] \) such that
\begin{equation}  \label{eq:interp_error}
f(0) - p_n(0) =  \frac{(-1)^{n} f^{(n+1)} (\xi)}{(n+1)!} \prod_{j=0}^{n} x_j.
\end{equation}

To derive explicit bounds, we assume that \(f\) satisfies the following regularity property: for some constants \(C\) and \(M > 0\), for all \(n \ge 0\),
\begin{equation}\label{eq:gevrey}
    \abs{f^{(n)}(x)} \le C \, M^n \quad \forall\, x \in [1,B].
\end{equation}
This condition is verified for the function in \cref{eq:f-definition} arising in QEM (Appendix \ref{app:exp_value_bound}). Condition \eqref{eq:gevrey} can be interpreted as a Gevrey class with \(\sigma = 0\).

\medskip

By substituting \eqref{eq:gevrey} into \eqref{eq:interp_error}, we obtain the following bound on the approximation error:
\begin{equation} \label{eq:expectation_error_bound}
     \abs{f(0) - p_n(0) } \le C \,\frac{M^{n+1}}{(n+1)!} \prod_{j=0}^{n} x_j.
\end{equation}
Hence, the error in approximating \(f(0)\) is influenced by the arrangement of the nodes \( x_j \). We will analyze this error by considering both equidistant and Chebyshev nodes.

\medskip

First, let us consider the case where all interpolation points are equidistant, i.e.,
\begin{equation}\label{xj-unif}
    x_j = 1 + jh, 
\end{equation}
where \( h := \frac{B - 1}{n} \) denotes the spacing.

\begin{theorem} \label{thm:equidist_error} 
Let $\{x_j\}_{j = 0}^{n}$ be a set of equidistant nodes given by \eqref{xj-unif}, and let \( p_n(x) \) be the degree-\( n \) polynomial interpolating the function \( f \) at these points. Under assumption \eqref{eq:gevrey}, for \( 0 < \varepsilon < e^{-e} \), if \( M \le B^{-\frac{B}{B-1}} \) and \( n = \Omega\bigl(\frac{\log(1/\varepsilon)}{ \sqrt{\log \log (1/\varepsilon)}}\bigr), \) then 
\(\lvert f(0) - p_n(0) \rvert < \varepsilon.\)
\end{theorem}

The proof is included in \cref{proof:equidist_error}. This result shows that, under the Gevrey-type bound on \( f \) in \eqref{eq:gevrey}, the interpolating polynomial \( p_n(x) \) approximates \( f \) at \( x = 0 \) with an error less than \( \varepsilon \).
Moreover, the theorem provides an optimistic estimate: to make the bias in \eqref{extrap-pn0-samp} sufficiently small, the number of interpolation points (and thus the number of circuits in QEM) scales only logarithmically with the precision~$\varepsilon$.

\medskip

Another widely recommended strategy in polynomial interpolation is to use \textit{Chebyshev nodes} \cite{Trefethen_approx}, which helps mitigate issues such as the Runge phenomenon \cite{Runge}. For a positive integer \(n\), the Chebyshev nodes of the first kind in the open interval \((-1,1)\) are given by
\begin{equation}
y_k = \cos \ \! \Bigl( \frac{(2k+1)\pi}{2(n+1)}\Bigr), 
\quad k=0,1,\dots,n,
\end{equation}
which are the roots of the Chebyshev polynomial \( T_{n+1}(y) \) of the first kind and degree \(n+1\) \cite{Trefethen_approx}.

To define Chebyshev nodes on the interval \([1,B]\), we use the linear bijection \( \phi: [-1,1] \to [1,B]\),
\begin{equation} \label{Chebyshev_shift}
     \phi(y) := \frac{B-1}{2}\,y + \frac{B+1}{2}.
\end{equation}
Consequently, the shifted Chebyshev polynomial on \([1,B]\) is
\begin{equation} \label{newT}
    \widetilde{T}_n (x) := T_n\!\Bigl(\phi^{-1}(x)\Bigr)
= T_{n}\!\Bigl(\frac{2(x-1)}{B-1} - 1\Bigr).
\end{equation}
Following the standard result for shifted nodes \cite{Trefethen_approx}, for \(\{x_j\}_{j=0}^{n} \subseteq [1,B]\) we obtain
\begin{equation}  \label{roots_Chebyshev_equivalence}
    \prod_{j=0}^{n} \bigl(x - x_j \bigr)
    = 2 \Bigl(\frac{B-1}{4}\Bigr)^{n+1} \,\widetilde{T}_{n+1}(x),
\end{equation}
which simplifies the interpolation error \eqref{eq:interp_error}. For convenience, we define the parameter
\begin{equation}\label{def:kappa}
    \kappa := \frac{\sqrt{B}+ 1}{\sqrt{B}-1}.
\end{equation}

\begin{theorem} \label{thm:Cheby_error_theorem}
Let $\{x_j\}_{j = 0}^{n}$ be a set of Chebyshev nodes in the interval \([1,B]\) with \(B > 1\), and let \( p_n(x) \) be a degree-\(n\) polynomial interpolating the function \( f \) at these points. Under assumption \eqref{eq:gevrey}, for \( 0 < \varepsilon < e^{-e} \), if \( M \leq \tfrac{4}{(B-1)e\kappa^2}\) and 
\(
n = \Omega\!\Bigl(\frac{\log(1/\varepsilon)}{\sqrt{\log\log(1/\varepsilon)}}\Bigr)
\), then 
\(
\lvert f(0) - p_n(0) \rvert < \varepsilon.
\)
\end{theorem}

For a detailed proof, see \cref{proof:Cheby_error_theorem}. Both Theorems \ref{thm:equidist_error} and \ref{thm:Cheby_error_theorem} show that \(n\) scales logarithmically in \(\varepsilon\). Comparing the assumptions in these two theorems, note that for large \(B\), \(B/(B-1) \approx 1\). Hence, for large \(B\), for equidistant nodes we have \(M \le B^{-1}\), which is slightly more restrictive than the corresponding bound for Chebyshev nodes. On the other hand, the bounds for \(M\) in both theorems can be relaxed if we consider a smaller interval for \(\varepsilon\). We elaborate on this point in Remark 1 of \cref{proof:Cheby_error_theorem}.

\bigskip

We now turn our attention to the variance of the estimator \( p_n(0) \) in \eqref{extrap-pn0-samp}. Since \( p_n(0) \) is a linear combination of independent random variables \( \hat{f}(x_j) \), we can compute its variance using the variance rule for a weighted sum:
\begin{equation} \label{eq:variance}
    \mathrm{Var}[p_n(0)] 
    = \mathrm{Var}\!\biggl[\sum_{j=0}^{n} \gamma_j \hat{f}(x_j)\biggr] 
    = \sum_{j=0}^{n} \gamma_j^2\,\mathrm{Var}\bigl[\hat{f}(x_j)\bigr],
\end{equation}
where \(\mathrm{Var}[\hat{f}(x_j)] = \tfrac{\sigma^{2}_j}{N_j}\) is the variance of the sample mean, and \(N_j\) is the number of samples taken at node \( x_j \).

For simplicity, assume \(\sigma_j \le \sigma\) for all \(j\), and let the number of samples be uniform across all nodes, i.e., \(N_j = N_S\) for every \( j \). Then \(\mathrm{Var}[\hat{f}(x_j)] \le \tfrac{\sigma^2}{N_S}\) for all \(j\). By applying the Cauchy–Schwarz inequality to \eqref{eq:variance}, we obtain
\[
    \mathrm{Var}[p_n(0)] 
    \;\le\; \frac{\sigma^2}{N_S}\sum_{j=0}^{n}\gamma_j^2
    \;\le\; \frac{\sigma^2}{N_S}\Bigl(\sum_{j=0}^{n}\lvert\gamma_j\rvert\Bigr)^2.
\]

To quantify the statistical error, we define the \(\ell_1\)-norm of the coefficients \(\gamma_j\):
\begin{equation}
    \|\bm{\gamma}\|_1 
    \;:=\; \sum_{j=0}^{n}\lvert\gamma_j\rvert.
\end{equation}
We will establish bounds on \(\|\bm{\gamma}\|_1\) and then incorporate these results into a Hoeffding bound at the end of this section.

\begin{theorem} \label{equidist_gamma_sum_bound}
    Let $\{x_j\}_{j = 0}^{n}$ be a set of equidistant nodes according to \eqref{xj-unif}. Then, the following holds:
    \begin{equation}\label{eq:gamma_sum}
    \norm{\bm \gamma}_1 = O\left(\left(\frac{2Be}{B-1}\right)^n\right).
    \end{equation}
\end{theorem}
 
A rigorous proof is provided in \cref{proof:equidist_gamma_sum_bound}. In numerical analysis, 
\(\norm{\bm \gamma}_1\) can be interpreted as a condition number that quantifies the sensitivity 
of the extrapolation to perturbations in the function values. Its exponential dependence on the 
number of interpolation points is well-known \cite{gautschi1987lower}. 

\medskip

Next, we analyze how choosing Chebyshev nodes affects \(\norm{\bm \gamma}_1\).

\begin{theorem} \label{thm:Cheby_gamma_sum_bound}
    Let $\{x_j\}_{j = 0}^{n}$ be the Chebyshev nodes in the interval \([1,B]\), where \( B > 1 \). Then, we have the bound:
\begin{equation}\label{cheby-gam-sum} 
    \norm{\bm \gamma}_1 = O\left( \kappa^{2n}  \right),
\end{equation}
where $\kappa$ is defined in \eqref{def:kappa}.

\end{theorem}

The detailed proof is provided in \cref{proof:Cheby_gamma_sum_bound}. Observe that \(2B \geq B + 1\), which implies that the bound for \(\norm{\bm \gamma}_1\) is smaller for Chebyshev nodes than for the equidistant nodes in \cref{equidist_gamma_sum_bound}. This aligns with Runge's phenomenon \cite{Runge}, since Chebyshev nodes reduce oscillations in polynomial interpolation, especially near the endpoints of the interpolation interval.

\bigskip

The analysis of the approximation and statistical errors above allow us to identify the sampling complexity of the ZNE method. We do so, by bounding the sampling complexity in the following theorem.

\begin{theorem} \label{thm:Hoeffding_Richardson}
Let \(f\), and \(p_{n}(x)\) be as in \cref{eq:f-definition} and \cref{def:extrap-pn} respectively.  Under assumption \eqref{eq:gevrey}, \( \forall\, 0 < \varepsilon < e^{-e}, 0 < \delta < 1\), the error bound \(\abs{\hat{p}_n(0) - f(0)} < \varepsilon\) holds with probability at least \(1 - \delta\)  if \(n = \Omega \left( \log \frac{1}{\varepsilon} \right)\), and and the number of samples \( N_S \) satisfies either one of the following conditions:
\begin{enumerate}
    \item $\{x_j\}_{j = 0}^{n}$ is a set of equidistant nodes according to \eqref{xj-unif}, \( M \leq B^{-B/(B-1)} \), and
    \[
        N_S = \Omega\left(\alpha^2  \varepsilon^{-(2+ 4Be/(B-1))}\log \frac{2}{\delta}\right).
    \]

    \item $\{x_j\}_{j = 0}^{n}$ is a set of Chebyshev nodes in the interval \([1,B] \), \(M \leq 4/((B-1) e \kappa^2)\),  and 
    \[
        N_S = \Omega\left(\alpha^2  \varepsilon^{-(2 + 4 \log \kappa)}\log \frac{2}{\delta}\right),
    \]
\end{enumerate}
where \(\kappa\) is defined in \eqref{def:kappa}.
\end{theorem}

 Note that we can replace \(B\) with \(((\kappa+1)/(\kappa - 1))^2\), as we have in \textit{Informal Theorem \ref{thm:info_1}}. We refer you to \cref{proof:Hoeffding_Richardson}  for the complete proof. Theorem \ref{thm:Hoeffding_Richardson} gives us a clear bound for the required number of circuits (nodes) and number of shots per circuit to reach a small enough error using Richardson's extrapolation.
 
\medskip 
 More specifically, in the context of QEM, $M$ in \cref{eq:gevrey} corresponds to the bound in \eqref{interpolation_bound}, as explained in \cref{app:exp_value_bound}. That is, \(M = \lambda_0 l T \). Specifically, $T$ is  the evolution time, and usually, large $T$  requires deeper circuits.   $\lambda_0$ is the base noise level and is associated with the error rate of the quantum device. Meanwhile, $l$ is a bound of the Lindblad noise operator, which is separate from  $\lambda_0$. $l$ is proportional to the number of qubits.  In light of the conditions on $M$ in the above theorem, the bounds for the number of data points ($n$) and the number of samples ($N_S$) are applicable for short circuit depth or low noise level.

 Meanwhile, if we fix the precision $\varepsilon,$ and consider deeper circuits, i.e., $M\gg 1,$ then, as elaborated in the analysis in \cref{Rem:largeM} in \cref{appendix-bia}, we find that $n= \Omega\left(M \varepsilon^{-(1/M)}\right).  $ Thus, according to \cref{equidist_gamma_sum_bound} and \cref{thm:Cheby_gamma_sum_bound}, the variance and the sample complexity will scale exponentially with $M$. This is consistent with the bounds in \cite{Universal_sampling_lower_bounds_QEM,takagi2022fundamental, Quek_tighter_bounds_2024}.

\begin{remark}[Importance Sampling] When the number of samples across all nodes is different, we can reduce the variance even further.
Since the coefficients in the extrapolation are known, the number of samples can be adjusted. Assume independent shot noise with
$\mathrm{Var}[\hat f(\lambda_i)] = \sigma_i^2/N_i$. With a fixed total budget $N_{\mathrm{tot}}=\sum_i N_i$, we can re-examine the total variance,
\begin{equation}
\begin{aligned}
&    \mathrm{Var}[\hat f(0)] = \sum_{i=0}^n \frac{\gamma_i^2 \sigma_i^2}{N_i} 
\quad \\ 
&\Rightarrow\quad
\min_{\{N_i\}} \sum_{i} \frac{\gamma_i^2 \sigma_i^2}{N_i}\ \ \text{s.t. }\sum_i N_i=N_{\mathrm{tot}}.
\end{aligned}
\end{equation}
By Lagrange multipliers, the optimal allocation is
\begin{equation}
N_i^\star \; = \; \frac{|\gamma_i|\,\sigma_i}{\sqrt{\lambda}},
\qquad
N_i^\star \;=\; N_{\mathrm{tot}}\,
\frac{|\gamma_i|\sigma_i}{\sum_j |\gamma_j|\sigma_j},
\end{equation}
yielding the minimum variance
\begin{equation}
\mathrm{Var}_{\min}[\hat f(0)] \;=\; \frac{\bigl(\sum_i |\gamma_i|\sigma_i\bigr)^2}{N_{\mathrm{tot}}}.
\end{equation}
Thus, unequal sampling strictly improves variance when $|\gamma_i|\sigma_i$ are non-uniform. Practically, allocating more shots to data points with large variance weight \(\abs{\gamma_i}\sigma_i\) minimizes the variance of the final Richardson estimator.
 \end{remark}

\section{Zero-Noise Extrapolation via Least Squares} \label{sec:least squares}

In this section, we investigate whether polynomials of lower degrees (i.e., \(m < n\)) can yield more robust estimates. Specifically, we construct a polynomial \(p_m\) of degree at most \(m\) to fit \(n+1\) nodes. Note that the discussion in the previous section corresponds to the special case \(m = n\). In statistical learning \cite{hastie2009elements}, the scenario \(m < n\) is typically formulated as a least-squares problem:
 \begin{equation}\label{eq:least_sq_statement}
    p_m(x)  = \displaystyle \arg \min_{p \in \mathcal{P}_m} \|p(x) -  f(x)\|^2_2,
\end{equation}
where \(\mathcal{P}_m\) is the space of polynomials of degree at most \(m\). 
Also known as polynomial regression, this method uses a polynomial degree significantly smaller than \(n\), effectively preventing overfitting. For extrapolation problems, it has been thoroughly investigated in \cite{Townsend}. Under the assumption that the function \(f(x)\) is analytic within a Bernstein ellipse. In particular, their analysis suggests that, under the oversampling condition
\( m \le \tfrac{1}{2}\sqrt{n},\)
the method achieves an asymptotically optimal extrapolant. However, in our setting, we do not necessarily assume analyticity in a Bernstein ellipse (although a subtle connection can be made \cite{demanet2010chebyshev}). Moreover, their analysis focuses on the interval \([-1,1]\) and extrapolates to \(-1\), with an error bound that benefits from the fact that the Chebyshev polynomials are still bounded by 1. In contrast, our noise scaling factors, due to the limitation of gate folding, must satisfy \(1 \leq x_i \leq B\). Consequently, the extrapolation point \(x=0\) may not lie within a scaled Chebyshev interval containing \([1,B]\), unless $B\geq \frac{16}{n^2 \pi^2}$, which is too restrictive. Another possibility, as suggested in \cite{cai2021multi}, is to combine ZNE with the probabilistic cancellation method that can reach a noise level with a scaling factor $x<1$.  

 In the following discussions,  we assume the data points in the least squares are \(n+1\) Chebyshev nodes in the interval \([1, B]\).
For better numerical stability and convenience of the analysis, we use the Chebyshev polynomial basis to expand \(f(x)\). Toward this end, let \(\{\tau_k(y)\}_{k=0}^{m}\) be the rescaled Chebyshev polynomials on the interval \([-1,1]\), that is
\begin{equation} \label{rescaled_Chebypoly}
    \tau_k(y) := \begin{cases}
        \sqrt{\frac{1}{n+1}} T_0(y), & k = 0, \\
        \sqrt{\frac{2}{n+1}} T_k(y), & k = 1, 2, \dots, m.
    \end{cases}
\end{equation}

For the purpose of extrapolating the function \(f\) to \(0\), we need to transform the Chebyshev polynomials from $[-1,1]$ to $[1,B]$.
 Again, we can use transform \eqref{newT} to shift the nodes:  Let
 \begin{equation}\label{tau-tilde}
     \tilde{\tau}_k(x) := \tau_k(\phi^{-1}(x)), 
 \end{equation}
 be the \(k\)-th shifted rescaled Chebyshev polynomial.  By orthonormality of \(\tilde{\tau}_i\)'s at Chebyshev points \cite{Cheby-poly-book,boyd2001chebyshev}, we have
 
\begin{equation} \label{cheby-orthonormal}
    \displaystyle \sum_{k=0}^{n} \tilde{\tau}_i(x_k) \tilde{\tau}_j(x_k) = \delta_{ij}, \quad  0 \leq i , j \leq m,
\end{equation}
where \(x_k\) is the \(k\)-th  Chebyshev node in $[1,B]$. We define the vector  \(\mathbf{x}= (x_0, x_1, \cdots, x_n)\) to combine the \(n+1\) Chebyshev nodes on $[1,B] $. Now we can expand \(f\) using the Chebyshev series on the interval \([1,B]\)
\begin{equation} \label{eq:f_cheby_series}
      f(x) = \sum_{k=0}^{\infty} a_k \tilde{\tau}_k(x).
 \end{equation}
The coefficients $a_k$ can be expressed as the integrals of $f(x) \tilde{\tau}_k(x)$. Meanwhile, we can write our approximating polynomial in \eqref{eq:least_sq_statement} in terms of the Chebyshev series as well.
\begin{equation} \label{eq:least_sq_approx}
    p_m(x) =  \sum_{k = 0}^{m} c_k \tilde{\tau}_k(x),
\end{equation}
where \(c_k\) is the  $k$-th coordinate of \(\mathbf{c}\), the vector of coefficients that minimizes the \(\ell_2\)-norm in the least squares approximation \eqref{eq:least_sq_statement}. 

The least squares problem can be expressed in the matrix form. Let \(\tilde{V}(\mathbf{x})\) denote the \((n+1) \times (m+1)\) shifted rescaled Chebyshev-Vandermonde matrix \cite{Townsend},
\begin{equation}
    \tilde{V}(\mathbf{x}) = \begin{bmatrix}
        \tilde{\tau}_0(x_0) & \cdots & \tilde{\tau}_m(x_0) \\
        \cdots  & \ddots & \vdots \\
        \tilde{\tau}_0(x_n) & \cdots & \tilde{\tau}_m(x_n)
    \end{bmatrix}.
\end{equation}
Further we define \(\mathbf{y} := (f(x_0), f(x_1), \cdots, f(x_n))
 \) as the function values. Now the problem in \eqref{eq:least_sq_statement} is transformed to the following  algebraic form, 
\begin{equation}
        \displaystyle \min_{\mathbf{c}} \norm{ \tilde{V}(\mathbf{x}) \ \mathbf{c} - \mathbf{y}}_2^2.
\end{equation}

In particular, the normal equation reads, 
\begin{equation}
    \tilde{V}(\mathbf{x})^{T} \tilde{V}(\mathbf{x}) \mathbf{c} = 
    \tilde{V}(\mathbf{x})^{T} \mathbf{y} .
\end{equation}
 By orthonormality \eqref{cheby-orthonormal}, \(
\tilde{V}(\mathbf{x})^{T} \tilde{V}(\mathbf{x}) = I_{(m+1) \times (m+1)}
\). Therefore, \(
    \mathbf{c} = 
     \tilde{V}(\mathbf{x})^{T} \mathbf{y}
\), which we can write out explicitly,
\begin{equation}\label{c_expanded}
    \mathbf{c} = 
    \begin{pmatrix}
 \displaystyle \sum_{i=0}^{n} \tilde{\tau}_0(x_i) y_i \\
 \displaystyle \sum_{i=0}^{n} \tilde{\tau}_1(x_i) y_i \\
\vdots \\
 \displaystyle \sum_{i=0}^{n} \tilde{\tau}_m(x_i) y_i
\end{pmatrix}.
\end{equation}

To identify the degrees of the polynomial \(p_m\) that ensure the approximation error remains below \(\varepsilon\), we establish the following theorem,

\begin{theorem} \label{thm:bias_bound_lsa}
   Let \(f(x)\) be under condition \eqref{eq:gevrey}, for \(M < 1\), and let \(p_m(x)\) denote a polynomial of degree \(m \leq n\) obtained through the least squares approximation of \(f\), interpolating the function \( f \) at a set of  \(n\) Chebyshev nodes in the interval \([1,B], B > 1\). For every \(\varepsilon > 0\),  \(\abs{p_m(0) - f(0)} < \varepsilon\) holds, provided that  \( \kappa \leq M^{-\mu m/2n}\) for some \(\mu \in (0,1)\), and 
\begin{equation}\label{bound-m}
    m = \Omega \left( \frac{ \log (C^{'}/\varepsilon)
   }{(1-\mu) \log (1/M)} \right), 
\end{equation}
{where} $ C^{'} = \frac{2 (B-1) C M}{ \pi} (\frac{1}{1 -M \kappa^2} + \frac{1}{1 - M}). $ 
\end{theorem}

See \cref{proof:bias_bound_lsa} for the detailed proof. 

We note that  in Theorem \ref{thm:Cheby_error_theorem}, \(\kappa^2 \leq \frac{4}{(B-1)Me}\). However, as \(B \to \infty\), \(\kappa^2 \to 1\), and \(\frac{4}{(B-1)Me} \to 0\), from which we conclude that the condition regarding $M$ in Theorem \ref{thm:Cheby_error_theorem} is more restrictive than in Theorem \ref{thm:bias_bound_lsa}.

\medskip 

The above theorem completes our analysis of the error (bias) in the least squares method. We once again return to variance analysis. We first note that
\(
    c_k  = \displaystyle \sum_{i=0}^{n} \tilde{\tau}_k(x_i) y_i
\)
in
\(
    p_{m}(0) = \displaystyle \sum_{k = 0}^{m} c_k \tilde{\tau}_k(0)
\). Combining them together, we find that \(\displaystyle p_{m}(0) = \sum_{i = 0}^{n} y_i \sum_{k=0}^{m} \tilde{\tau}_k(x_i) \tilde{\tau}_k(0) = \sum_{i=0}^n \gamma_i y_i,
\)
where,
\begin{equation}\label{tau2gamma}
    \gamma_i :=  \displaystyle \sum_{k=0}^{m} \tilde{\tau}_k(x_i) \tilde{\tau}_k(0).
\end{equation}
Consequently,  $p_m(0)$ is expressed as a linear combination of  \(f(x_i)\). Estimating the statistical error is again reduced to 
 estimating \(  \displaystyle \norm{\bm \gamma}_1=\sum_{i = 0}^{n}  \abs{\gamma_i}
\).

\begin{theorem} \label{thm:LSA_gamma_sum}
     Let $\{x_j\}_{j = 0}^{n}$ be Chebyshev nodes in the interval \([1, B]\), \( B > 1 \),  \(\{\tilde{\tau}_k(x)\}_{k=0}^{m}\) be shifted normalized Chebyshev polynomials of the first kind on \([1,B]\), and \(  \displaystyle \norm{\bm \gamma}_1=\sum_{i = 0}^{n}  \abs{\gamma_i}
\), where \( \gamma_i :=  \displaystyle \sum_{k=0}^{m} \tilde{\tau}_k(x_i) \tilde{\tau}_k(0)\). Then,
    \begin{equation}
        \norm{\bm \gamma}_1= O\left(\kappa^{2m}\right). 
    \end{equation}
\end{theorem}

The proof of the above theorem can be found in \cref{proof:LSA_gamma_sum}. By comparing the bound for \( \norm{\bm \gamma}_1\) in \cref{thm:Cheby_gamma_sum_bound} and \cref{thm:LSA_gamma_sum}, we see that the least squares approximation method results in a smaller variance because $m<n$, which 
could reduce the statistical error and overfitting.

With the prior analysis of the bias and variance in the least squares method, we present the following theorem that calculates the number of shots needed to keep the bias below the threshold \(\varepsilon\).

\begin{theorem} \label{thm:Hoeffding_lsa_thm}
Let \(f\), and \(p_{n}(x)\) be as in \cref{eq:f-definition} and \cref{eq:least_sq_approx} respectively.  Under assumption \eqref{eq:gevrey}, \( \forall\, 0 < \varepsilon , 0 < \delta < 1\), the error bound \(\abs{\hat{p}_n(0) - f(0)} < \varepsilon\) holds with probability at least \(1 - \delta\)  if  \( \kappa \leq M^{-\mu m/2n}\) for \(\mu \in (0,1)\), \(m = \Omega \left( \frac{ \log (C^{'}/\varepsilon)
   }{(1-\mu) \log (1/M)} \right)\), and the number of samples \( N_S \) satisfies the following condition:
\[
N_S = \Omega \left((2 \alpha^2 \kappa^{\log C'/(1-\mu)} \log \frac{2}{\delta}) \varepsilon^{-2 -\log \kappa/(1-\mu)}\right),
\]
 where \(C^{'} = \frac{2 (B-1) C M}{ \pi} (\frac{1}{1 -M \kappa^2} + \frac{1}{1 - M}) \).
\end{theorem}

For the complete proof, refer to \cref{proof:Hoeffding_lsa_thm}. This concludes our analysis 
of polynomial approximations for ZNE.

\section{Incorporating the Trotter Error} \label{sec:Trotter_error}
The analysis in the preceding sections has primarily focused on mitigating circuit errors 
(also known as physical errors \cite{EndoMitigAlgoError}). Another intriguing direction 
is to employ the same strategy to reduce numerical errors in a given algorithm, referred to 
as algorithmic errors. This approach has been explored in \cite{ImprovedAccuracyforTrotterSim,watson2024exponentially}. 
A key application is the Trotter-Suzuki decomposition in Hamiltonian simulations.

More specifically, the Trotter-Suzuki decomposition (also known as the Trotter splitting method) 
is a widely used technique for approximating the time evolution operator 
\(\displaystyle U(t) = e^{-iHt}\), where \(H\) is a Hamiltonian composed of multiple terms, 
\begin{equation}
\displaystyle H = \sum_{k=1}^{K} H_k ,
\end{equation}
each of which is simpler to simulate individually. 
A straightforward way to implement this decomposition is the first-order Trotterization, 
in which the time evolution operator over a time \(t\) is approximated using small time steps of size \(\tau\) as:
\begin{equation}
  e^{-iHt} \;\approx\; S_1(t) \;:=\; \Bigl(\prod_{k=1}^{K} e^{-iH_k \tau}\Bigr)^N,
\end{equation}
where \(\tau = t / N\) is the time step size, and \(N\) is the number of Trotter steps used 
to evolve over the total simulation time \(t\).
Although straightforward, the first-order decomposition introduces an approximation error that scales with \(\tau^2\) at each step. Consequently, many higher-order Trotter-Suzuki expansions have been developed to achieve more accurate simulations \cite{Childs_Trotter_Theory_2021}. For example, the symmetric second-order Trotterization uses a symmetric splitting
\begin{equation}\label{symmetric-trotter}
\begin{aligned}
     e^{-iHt} \approx & S_2(t) := \Big( e^{-iH_1\tau/2} \dots e^{-iH_{K-1}\tau/2} \\ 
     & \times e^{-iH_K \tau} e^{-iH_{K-1} \tau/2} \dots e^{-iH_1 \tau/2} \Big)^N,
\end{aligned}
\end{equation}
leading to a reduced error that scales with \(\tau^3\) at each step. More generally, a \(p\)-th order Trotter-Suzuki decomposition approximates \(e^{-iHt \,}\) with a {Trotter error} that scales as \( \mathcal{O}(\tau^{p+1}) \).

Increasing the order of the decomposition lowers the error, 
but at the expense of a quantum circuit whose length scales with the decomposition complexity, thereby exacerbating gate errors (physical errors). While the algorithmic error is theoretically well-characterized \cite{Childs_Trotter_Theory_2021}, its interplay with physical noise during quantum computation has received less attention.   This introduces a subtle interplay between algorithmic and physical errors: using a finer time step size \(\tau\) or higher-order methods can reduce Trotter error but at the risk of higher physical error from deeper circuits.

In this section, we incorporate Trotter error into the error mitigation analysis by systematically evaluating the combined bias. Building on this analysis, we then determine the required number of samples to ensure that the total error remains below a chosen threshold \(\varepsilon\).

Recall that in Hamiltonian simulation, the ideal quantum dynamics follows the time-dependent Schr\"odinger equation, or equivalently, the Liouville von Neumann equation,
\begin{equation}  \label{eq: ideal}
    \frac{\partial}{\partial t} \rho (t) =  -i [H(t), \rho(t)]. 
\end{equation}
Meanwhile, when approximated by the symmetric Trotter splitting method of stepsize $	\tau$ in \cref{symmetric-trotter}, the approximate solution can be associated with a solution of  
 a modified equation, where the Hamiltonian is expressed as even powers of $\tau$ \cite{hairer2006geometric}, 
 \begin{equation}\label{eq:modified-tdse}
     \begin{aligned}  
    \frac{\partial}{\partial t} \rho (t, \tau) & =  -i [H, \rho(t, \tau)] -i 	\tau^2   [\widetilde{H}_2, \rho(t, 	\tau)]\\
    & -i \tau^4   [\widetilde{H}_4, \rho(t, \tau)] 
    + \cdots.  
    \end{aligned}
 \end{equation}

The additional terms, e.g., $\widetilde{H}_2$,  in the modified Hamiltonian, involve the commutators of the individual Hamiltonian terms in $H$.

To simultaneously include the circuit noise, we consider a modified 
Lindblad equation by adding the dissipation term \(\lambda\). Now the density matrix \(\rho (t, \tau, \lambda)\) depends on the noise term \(\lambda\) as well,
\begin{equation}  \label{eq: modified-lindblad}
\begin{aligned}
        &\frac{\partial}{\partial t} \rho (t, \tau, \lambda)   \\
        & =  -i [H, \rho(t, \tau, \lambda)] + \lambda \mathcal{L}(\rho(t, \tau, \lambda)) \\ 
        &-i \tau^2   [\widetilde{H}_2, \rho(t, \tau, \lambda)] -i \tau^4   [\widetilde{H}_4, \rho(t, \tau, \lambda)] - \cdots   \\
        & =  -i [H, \rho(t, \tau, \lambda)] + \lambda \mathcal{L}(\rho(t, \tau, \lambda)) \\ & - i \sum_{\ell\geq 1} \tau^{2\ell}   [\widetilde{H}_{2\ell}, \rho(t, \tau, \lambda)].
\end{aligned}
\end{equation}
A similar expansion for \(\rho(t,\tau)\) has been derived previously \cite{watson2024randomly}. However, the expansion above considers both the Trotter step size \(\tau\), and the noise level \(\lambda\).\\
Rather than generating samples in a two-dimensional parameter space, we simultaneously scale the physical noise, and the step size, reducing the problem to a one-dimensional extrapolation of the expectation value of an observable.\\
In light of the scaling of the terms, we make a simple choice for the circuit noise, i.e., $\lambda = O( \tau^2 ) $ then we can define \(\rho(t, \lambda):= \rho (t, \sqrt{\lambda}, \lambda)\), and examine the derivatives,
\begin{equation}\label{Lamk}
    \Lambda_k (t) := \frac{\partial^k}{\partial \lambda^k} \rho (t, \lambda) , \,\; k \geq 1.
\end{equation}

For convenience, we write the modified equation \eqref{eq: modified-lindblad} as 
\(\frac{\partial}{\partial t} \rho (t, \lambda)= \mathcal{L}  \rho (t, \lambda), \)
with the operator $\mathcal{L}$ expanded in terms of the powers of $\lambda$,
\begin{equation}\label{eq:L-expand}
    \mathcal{L} = \mathcal{L}_0 + \lambda \mathcal{L}_1 + \cdots + \lambda^k \mathcal{L}_k + \cdots. 
\end{equation}

To derive an explicit error bound, we make the assumption that 
\begin{equation}\label{Lj-bound}
    \norm{\mathcal{L}_j } \leq \theta^j, \quad \forall\, j\geq 0. 
\end{equation}
This is slightly different from the assumption in \cite[Lemma 2]{watson2024exponentially}.

Under this condition, we find that 
the bias can be controlled by properly increasing the number of nodes, as summarized below,
\begin{theorem} \label{thm:trotter_error}
     Let \(\rho\) be the density matrix satisfying the  modified Lindblad equation \eqref{eq: modified-lindblad}, and $f(\lambda) := \tr(\rho(t, \lambda) A)$. Also, let \(p_n(x)\) denote the polynomial interpolant of \(f\), where \(f\) and \(p_n\) coincide on Chebyshev nodes \(\{x_i\}_{i=0}^{n}\) in the interval \([1,B]\), \(B > 1\).
     If \(n = \Omega \left(\log \varepsilon/\log (\frac{(B-1) e \kappa^2 \theta}{4 (1 - \lambda \theta)})\right)\), then \(\abs{f(0) - p_n(0)} < \varepsilon\).
\end{theorem}

For the proof, we refer readers to \cref{proof:trotter_error}. It is worthwhile to mention that the resulting function $f$ no longer fulfills the condition \cref{eq:gevrey}. Instead, the bound for the $n$-th order derivative involves an additional factor $n!$. Thus it belongs to a Gevrey class with $\sigma=1$ (see \cref{gevrey-1}).    Now we can find a bound for the sampling complexity.

\begin{theorem} \label{thm:joint_Hoeffding}
    Let \(f(x)\), and \(p_n(x)\) be as in Theorem \ref{thm:trotter_error}, then, \(\abs{p_n(0) - f(0)} < \varepsilon\) with probability at least \(1 - \delta\) for \(0 < \varepsilon < e^{-e}, 0 < \delta < 1\), if \(n = \Omega \left(\log \varepsilon/\log (\frac{(B-1) e \kappa^2 \theta}{4 (1 - \lambda \theta)})\right)\), and 
    \[
    N_S = \Omega\left(\alpha^2  \varepsilon^{-(2+4 \log \kappa)}\log\frac{2}{\delta}\right).
    \]
\end{theorem}

\begin{proof}
    The proof is the same as the proof of Theorem \ref{thm:Hoeffding_Richardson}, with the exception that, instead of Theorem \ref{thm:Cheby_error_theorem}, we use Theorem \ref{thm:trotter_error}. Note that  \(n \geq \log \frac{1}{\varepsilon}\) because \( \frac{(B-1) e \kappa^2 \theta}{4 (1 - \lambda \theta)} < 0\). 
\end{proof}

\section{Numerical Experiments}\label{sec:Numerics}

In this section, we study the performance of some ZNE methods by numerically testing them on a transverse-field Ising model (TFIM). The TFIM is widely used in quantum computing and condensed matter physics to investigate quantum phase transitions and entanglement properties, owing to the rich interplay between inter-qubit interactions and an external transverse field. 
Specifically,  the Hamiltonian in our study is given by
\begin{equation}
    H = - J \sum_{i=0}^{L-2} Z_i Z_{i+1} \;-\; h \sum_{i=0}^{L-1} X_i,
\end{equation}
where \(X_i\) and \(Z_i\) are the spin-\(\tfrac{1}{2}\) Pauli \(X\) and \(Z\) operators, respectively, acting on the \(i\)-th qubit. The first term, \(-J \sum_{i=0}^{L-2} Z_i Z_{i+1}\), represents the interaction between neighboring qubits, promoting alignment of their spins along the \(Z\)-axis. The second term, \(-h \sum_{i=0}^{L-1} X_i\), accounts for the effect of a transverse magnetic field, which flips qubits between the \(\ket{0} \) and \(\ket{1} \) states. 

 In implementing the ZNE methods, we employ a five-qubit system (\(L = 5\)) with coupling and transverse field strengths set to \(J = 0.2\), and \(h = 1\). The qubits are initialized in the product state \(\ket{\psi(0)} = \ket{0}^{\otimes L}\). After evolving the system until time \(t\) under the TFIM Hamiltonian, we measure the observable \(A = X_1\), which corresponds to applying the Pauli-\(X\) operator on the first qubit.

To simulate the circuit noise, we implement the commonly used \emph{depolarizing channel}, which transforms an \(L\)-qubit density matrix \(\rho\) into 
\(
\Delta_\lambda(\rho) = (1 - \lambda)\rho + \lambda \,\frac{1}{2^L} I,
\)
where \(I\) is the identity operator on \(n\) qubits, and complete positivity requires \(0 \leq \lambda \leq 1 + \tfrac{1}{4^L - 1}\). We set the initial noise parameter to \(\lambda_0 = 0.02\) in our numerical experiments. 

All simulations reported in this section are performed using the \(\text{QasmSimulator}\) from IBM’s Qiskit \cite{Qiskit}, which classically simulates the dynamics of the TFIM Hamiltonian as a quantum circuit, under the settings mentioned above.

\begin{figure}[!htbp] 
    \centering
    \includegraphics[scale=0.25]{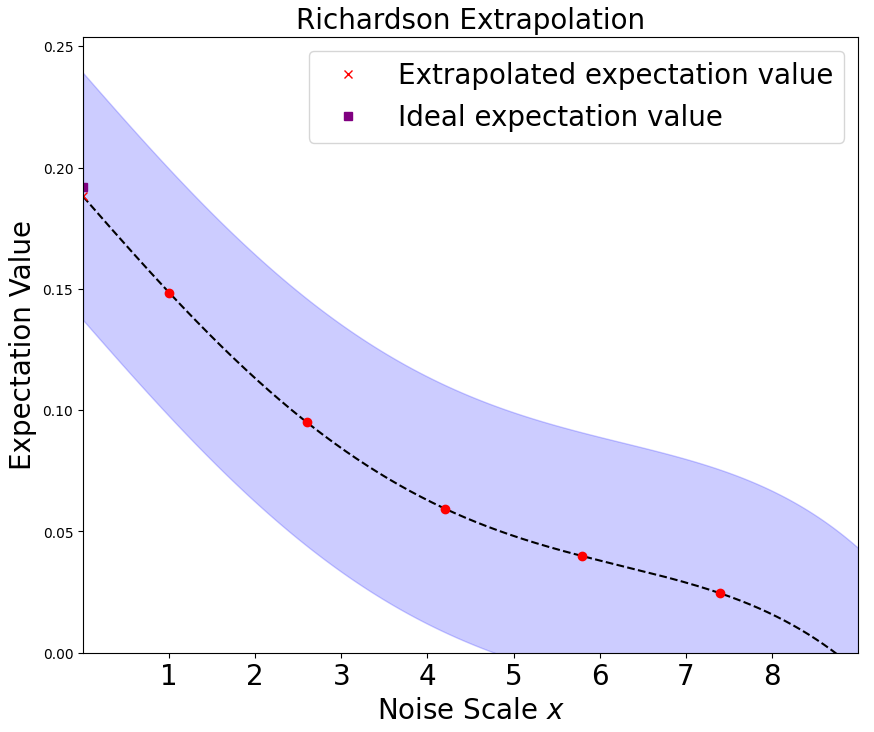}
    \caption{A simple illustration of Richardson's extrapolation method for quantum error mitigation using equidistant noise scales. The ideal zero-noise expectation value is \(0.191826\). In comparison, the extrapolated zero-noise expectation value is \(0.188129\).}
    \label{fig:basic_interpolation}
\end{figure}

We begin with a simple demonstration of Richardson extrapolation by running the system 
at eight distinct noise levels. We simulate each circuit using the second-order Trotter-Suzuki method. 
Specifically, we start with a noisy Trotterized circuit at an initial noise level \(\lambda_0\), 
and then apply gate folding to amplify the noise by factors \(x_i\) for seven additional circuits 
\cite{larose2022mitiq}. We select five equidistant nodes for the noise scales, 
each separated by a uniform interval of length~1, and run each circuit with \(N_S = 10^6\) shots.

Finally, we measure the Pauli-X operator \(X_1\) on each circuit and plot the resulting expectation 
values in Figure~\ref{fig:basic_interpolation}. We then extrapolate these values to \(x=0\) 
using a polynomial \(p_4\) of degree at most~4, effectively implementing Richardson extrapolation. 
As shown in Figure~\ref{fig:basic_interpolation}, the extrapolated expectation value is 
\(0.188129\), whereas the ideal noiseless circuit yields \(0.191826\). 
The shaded region in the figure depicts the standard deviation from shot noise, 
providing a visualization of the statistical confidence. 
Thus, Richardson extrapolation approximates the ideal value with an error of about \(0.017144\). 
Our goal in the following experiments is to investigate how different approximation methods 
affect this error.

\begin{figure}[!htbp] 
    \centering
    \includegraphics[scale=0.25]{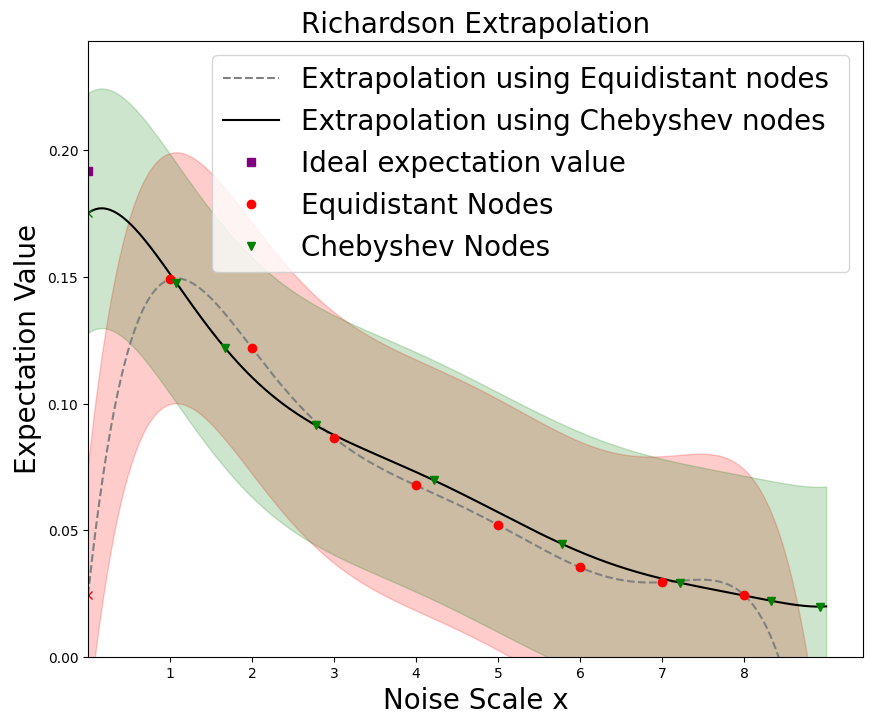}
    \caption{Richardson extrapolation using equidistant and Chebyshev nodes.  The ideal zero-noise expectation value is \(0.191826\), while the extrapolated expectation values using equidistant nodes and Cheyshev nodes are \(0.17516\), and \(0.0244759\), respectively.}
    \label{fig:interpolation_cheby_vs_equi}
\end{figure}

\paragraph{The distribution of the nodes. }
To test the effect of the distribution of extrapolation nodes in the context of ZNE, we run the circuit we previously defined, for \(8\) distinct equidistant noise scales, using Chebyshev nodes as the noise scales. In Figure \ref{fig:interpolation_cheby_vs_equi}, the red circles represent the expectation values we obtain from circuits in which the noise scales are equidistant nodes, while the green triangles represent their counterparts in which the scales are shifted Chebyshev nodes. 
From a comparison with the ideal expectation value, we can see that the approximation error is significantly smaller when using Chebyshev nodes. These results are also consistent with the observations from \cite{MichaelKrebsbachOptimizingRichardson}.

\paragraph{Extrapolation using least-squares.} Next we test the least-square approach in Section \ref{sec:least squares}. 
In this case, the approximating polynomial is of degree \(m < n\) for \(n\) sample nodes. Under the same circuit settings as the previous experiment, i.e., $n=8,$ we consider the same distributions of nodes. By changing the degree of the approximating polynomial to \(3\), we show the fitting polynomial and its extrapolation to $x=0$ in Figure \ref{fig:LSA_cheby_vs_equi}. 
Due to the lower degree of the polynomials, the least-squares polynomials show significantly less oscillations. More importantly, compared with the results in Figure \ref{fig:interpolation_cheby_vs_equi},  the results from least-squares in Figure \ref{fig:LSA_cheby_vs_equi} improve the approximation error for both equidistant and Chebyshev nodes.

\begin{figure}[!htbp] 
    \centering
    \includegraphics[scale=0.25]{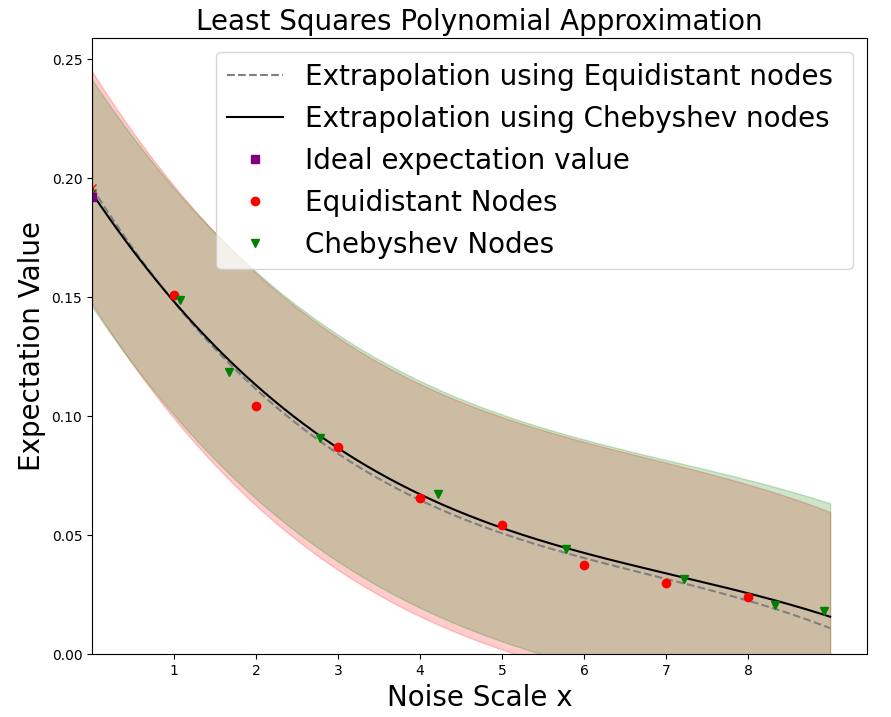}
    \caption{
   ZNE from the least-square approach \eqref{eq:least_sq_statement} using equidistant and Chebyshev nodes. The extrapolated expectation values using equidistant nodes and Cheyshev nodes are \(0.193625\) and \(0.195975\), respectively, while the ideal expectation value is \(0.191826\). }
    \label{fig:LSA_cheby_vs_equi}
\end{figure}

\paragraph{The choice of the degree of the approximating polynomial.} Next, under the previous settings, we fix $n=20$ Chebyshev nodes \( x_j \in [1,30] \), while varying the degree of the least-squares polynomial from \(m=2\) to \(m=20\). The results in Figure \ref{fig:bias_vs_degree_of_poly} suggest that as the polynomial degree \(m\) approaches the number of nodes (\(n = 20\)), the approximation error increases. In particular, the error is largest when \(m = n\). The minimum error is observed at \(m = 11\), which, in practice, can be identified using cross-validation. As expected, the error remains lower when the polynomial degree is less than \(n-1\), reinforcing that when $n\gg 1$, lower-degree polynomials that we obtain through least-squares are generally more suitable.

\begin{figure}[!htbp] 
    \centering
    \includegraphics[scale=0.3]{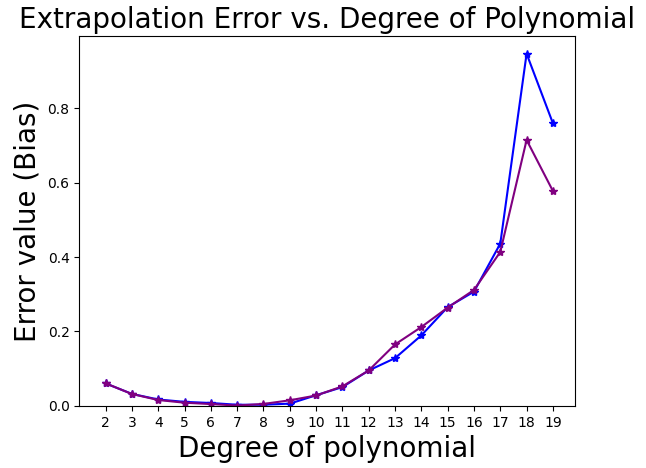}
    \caption{The change of extrapolation error with respect to the degree of the approximating polynomial in the least-square approach \eqref{eq:least_sq_statement}. }
    \label{fig:bias_vs_degree_of_poly}
\end{figure}

\paragraph{Extrapolation of Trotter error. } As mentioned in the introduction, the expectation value of an observable is influenced by both the physical error \(\lambda\) and the Trotter step size \(\tau\) in the Trotter-Suzuki algorithm. Therefore, the extrapolation methods we consider can be applied with respect to both \(\lambda\) and \(\tau\).

To isolate the effect of extrapolation with respect to the Trotter step size, we simulate a set of circuits without physical noise, i.e., $\lambda=0$. Specifically, we implement the second-order Trotter-Suzuki algorithm with varying step sizes on noiseless circuits. For a total evolution time of \(T = 2\), we choose the number of time steps \(N_T\) such that \(\tau = T/N_T\). The results are presented in Figure \ref{fig:only_trotter_error}, where the \(y\)-value of each blue point represents the expectation value of a circuit with a corresponding Trotter step size indicated along the \(x\)-axis. 

As before, we obtain the expectation values by measuring the Pauli \(X\) operator on the first qubit of each circuit, with each circuit executed one million times. The extrapolation to \(\tau = 0\) is performed using a least-squares fit with a polynomial of degree \(5\), yielding results that closely agree with the exact expectation value, which we compute directly via matrix exponentiation:
$\ket{\psi(t)}= e^{-itH} \ket{\psi(0)}.$ 
 More numerical results can be found in \cite{ImprovedAccuracyforTrotterSim}.

\begin{figure}[tph] 
    \begin{center}
    \includegraphics[scale=0.25]{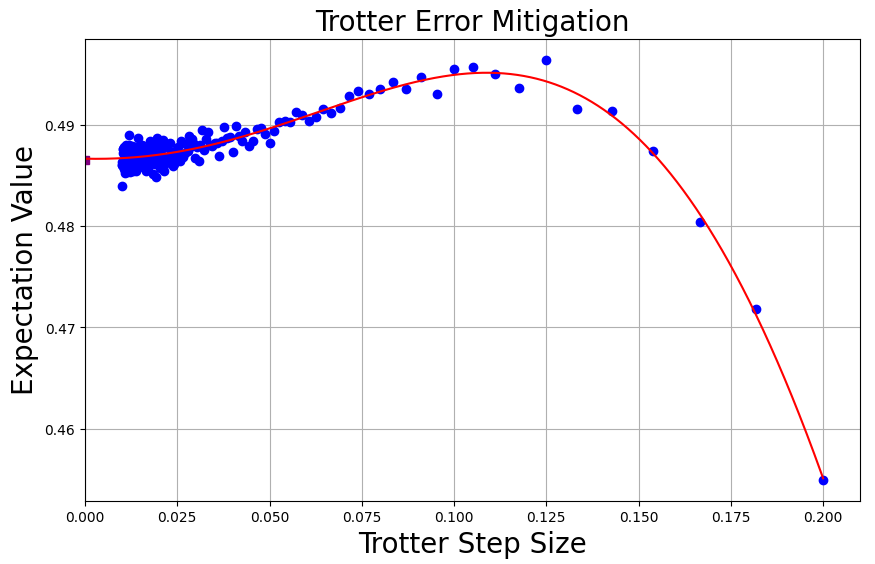}
    \end{center}
    \caption{Trotter error mitigation by changing the number of Trotter steps from $N_T=200$ (step size $0.01$) to $N_T = 10$ (step size $0.2$), for a total evolution time of $T=2$. The solid line indicates extrapolation using a polynomial of degree \(5\). The exact expectation value is \(0.48652\), while the extrapolated expectation value is \(0.48663\).}
    \label{fig:only_trotter_error}
\end{figure}

\paragraph{Simultaneous mitigation of circuit and algorithmic errors.}
The final experiment examines a scenario where both Trotter and physical errors are present. To achieve noise scales \( x \ge 1 \), we set 
This \(x = \frac{c}{\lambda_0}\,\tau^2\) 
for some positive constant \(c\). Each combination of step size and corresponding noise level is implemented as a separate circuit.

Here, we measure \(Z_1\), the Pauli-Z operator on the first qubit, because depolarizing noise tends to drive the expectation value of \(X_1\) to nearly zero, making \(Z_1\) a more informative choice for observing the effects of noise. The dots in Figure~\ref{fig:joint_exp_val} represent the measured expectation values, which appear to follow a roughly quadratic trend. By performing a least-squares fit with a degree-5 polynomial, we extrapolate these measurements to approximate the exact expectation value with satisfactory accuracy.

\begin{figure}[H] 
    \centering
    \includegraphics[scale=0.25]{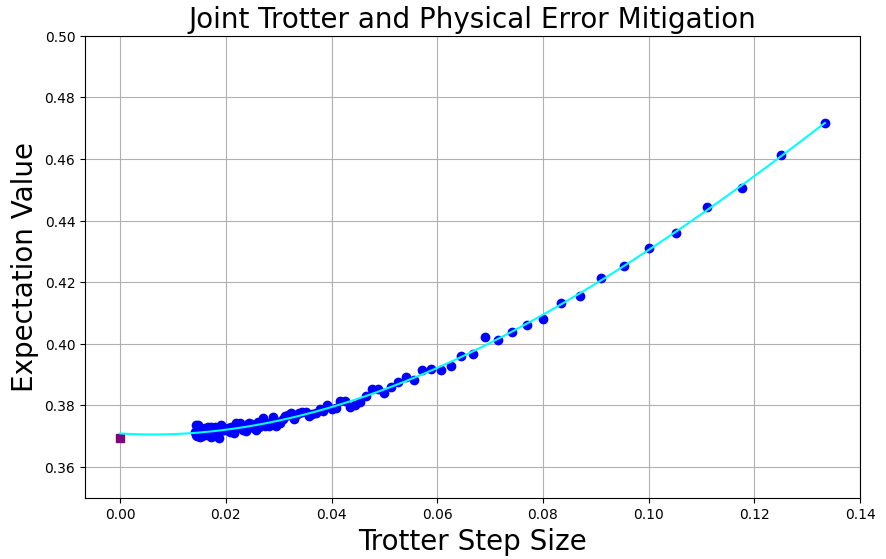}
    \caption{Joint mitigation of Trotter and physical error by varying the number of Trotter steps from \(15\) to \(141\), and letting the physical error be \(\lambda = c \tau^2\), where \(c = 100\), and \(\tau\) is the Trotter step size. The exact expectation value is \(0.3693\).}
    \label{fig:joint_exp_val}
\end{figure}

\section{Summary and Discussions} \label{sec:summary_discussions}
This paper provides a detailed analysis of the accuracy of zero-noise extrapolation (ZNE), 
a widely used quantum error-mitigation strategy. By exploiting properties of the function 
governing how circuit expectations depend on noise---alongside polynomial approximations---
we establish bounds on both the bias and variance of ZNE estimates. 
These bounds, in turn, lead to a sampling complexity that ensures accurate results 
with high probability.

Our analysis extends in two important directions. First, we introduce polynomial least-squares approximations that help prevent overfitting and curb the escalation of statistical error. Second, we propose a method that mitigates both algorithmic and circuit errors simultaneously by treating numerical and noise parameters jointly, and we provide rigorous error bounds for this unified approach.

In the broader realm of statistical learning, polynomials constitute only a narrow class of potential function choices. Therefore, it is a promising direction to investigate ZNE using more general function families. For instance, the works in \cite{endo2018practical,cai2021multi} adopted exponential functions that exhibited better performance. Extending our framework to such broader classes presents an intriguing avenue for future research.

Finally, the function properties used in our analysis are relatively crude, as they do not account for specific noise models or observables. Developing more specialized characterizations could yield tighter error bounds and reduce sampling overhead.

\emph{Note.} Shortly after the initial posting of our manuscript, another paper, Hakkaku \emph{et~al.}~\cite{hakkaku2025data}, appeared with ideas similar to those presented in Section~4. Both works aim at the 
\emph{joint} mitigation of physical and algorithmic (Trotter) errors with improved data efficiency, in order 
to avoid the variance growth that arises in separate extrapolation of physical and algorithmic noise. In particular, both approaches construct 
a one-dimensional extrapolation by tuning simultaneously the physical noise and Trotter steps, with results validated by numerical tests. One key difference is that our work considers general continuous noise described by Lindblad equation, rather than a specific noise channel, and our analysis involves estimating the smoothness of the expectations on the noise and step size parameters, which enables a rigorous analysis of the bias and variance trade-off.

\paragraph{Acknowledgments} This research is supported by the NSF Grants DMS-2111221 and CCF-2312456.

\addcontentsline{toc}{section}{References}
\bibliographystyle{quantum}  

\begin{thebibliography}{10}

\bibitem{preskill2018nisq}
John Preskill.
\newblock ``Quantum computing in the nisq era and beyond''.
\newblock \href{https://dx.doi.org/10.22331/q-2018-08-06-79}{Quantum {\bf 2},
  79}~(2018).

\bibitem{li2017efficient}
Ying Li and Simon~C Benjamin.
\newblock ``Efficient variational quantum simulator incorporating active error
  minimization''.
\newblock \href{https://dx.doi.org/10.1103/PhysRevX.7.021050}{Physical Review X
  {\bf 7}, 021050}~(2017).

\bibitem{KTemme_SBravyi_JGambetta}
Kristan Temme, Sergey Bravyi, and Jay~M Gambetta.
\newblock ``Error mitigation for short-depth quantum circuits''.
\newblock
  \href{https://dx.doi.org/https://doi.org/10.1103/PhysRevLett.119.180509}{Physical
  review letters {\bf 119}, 180509}~(2017).

\bibitem{surface_codes}
Austin~G Fowler, Matteo Mariantoni, John~M Martinis, and Andrew~N Cleland.
\newblock ``Surface codes: Towards practical large-scale quantum computation''.
\newblock \href{https://dx.doi.org/10.1103/PhysRevA.86.032324}{Physical Review
  A—Atomic, Molecular, and Optical Physics {\bf 86}, 032324}~(2012).

\bibitem{lindblad1976generators}
Goran Lindblad.
\newblock ``On the generators of quantum dynamical semigroups''.
\newblock
  \href{https://dx.doi.org/https://doi.org/10.1007/BF01608499}{Communications
  in mathematical physics {\bf 48}, 119--130}~(1976).

\bibitem{gorini1976completely}
Vittorio Gorini, Andrzej Kossakowski, and Ennackal Chandy~George Sudarshan.
\newblock ``Completely positive dynamical semigroups of n-level systems''.
\newblock \href{https://dx.doi.org/0.1063/1.522979}{Journal of Mathematical
  Physics {\bf 17}, 821--825}~(1976).

\bibitem{breuer2002theory}
Heinz-Peter Breuer and Francesco Petruccione.
\newblock ``The theory of open quantum systems''.
\newblock
  \href{https://dx.doi.org/10.1093/acprof:oso/9780199213900.001.0001}{OUP
  Oxford}. ~(2002).

\bibitem{EvanB_MiladM}
Evan Borras and Milad Marvian.
\newblock ``A quantum algorithm to simulate lindblad master equations''.
\newblock
  \href{https://dx.doi.org/10.48550/arXiv.2406.12748}{\href{https://doi.org/10.48550/arXiv.2406.12748}{arXiv
  preprint arXiv:2406.12748}}~(2024).

\bibitem{endo2018practical}
Suguru Endo, Simon~C Benjamin, and Ying Li.
\newblock ``Practical quantum error mitigation for near-future applications''.
\newblock \href{https://dx.doi.org/10.1103/PhysRevX.8.031027}{Physical Review X
  {\bf 8}, 031027}~(2018).

\bibitem{EndoMitigAlgoError}
Suguru Endo, Qi~Zhao, Ying Li, Simon Benjamin, and Xiao Yuan.
\newblock ``Mitigating algorithmic errors in a hamiltonian simulation''.
\newblock \href{https://dx.doi.org/10.1103/PhysRevA.99.012334}{Phys. Rev. A
  {\bf 99}, 012334}~(2019).

\bibitem{kurita2022synergetic}
Tomochika Kurita, Hammam Qassim, Masatoshi Ishii, Hirotaka Oshima, Shintaro
  Sato, and Joseph Emerson.
\newblock ``Synergetic quantum error mitigation by randomized compiling and
  zero-noise extrapolation for the variational quantum eigensolver''.
\newblock \href{https://dx.doi.org/10.22331/q-2023-11-20-1184}{{Quantum} {\bf
  7}, 1184}~(2023).

\bibitem{vazquez2022enhancing}
Almudena~Carrera Vazquez, Ralf Hiptmair, and Stefan Woerner.
\newblock ``Enhancing the quantum linear systems algorithm using richardson
  extrapolation''.
\newblock \href{https://dx.doi.org/10.48550/arXiv.2009.04484}{ACM Transactions
  on Quantum Computing {\bf 3}, 1--37}~(2022).

\bibitem{carrera2022extrapolation}
Almudena Carrera~Vazquez.
\newblock ``Extrapolation methods in quantum computing''.
\newblock \href{https://dx.doi.org/10.3929/ethz-b-000586831}{PhD thesis}.
\newblock ETH Zurich.
\newblock ~(2022).

\bibitem{hastie2009elements}
Trevor Hastie, Robert Tibshirani, Jerome~H Friedman, and Jerome~H Friedman.
\newblock ``The elements of statistical learning: data mining, inference, and
  prediction''.
\newblock \href{https://dx.doi.org/10.1007/978-0-387-84858-7}{Volume~2}.
\newblock Springer. ~(2009).

\bibitem{gautschi1987lower}
Walter Gautschi and Gabriele Inglese.
\newblock ``Lower bounds for the condition number of vandermonde matrices''.
\newblock
  \href{https://dx.doi.org/https://doi.org/10.1007/BF01398878}{Numerische
  Mathematik {\bf 52}, 241--250}~(1987).

\bibitem{li2006lower}
Ren-Cang Li.
\newblock ``Lower bounds for the condition number of a real confluent
  vandermonde matrix''.
\newblock
  \href{https://dx.doi.org/https://doi.org/10.1090/S0025-5718-06-01856-4}{Mathematics
  of computation {\bf 75}, 1987--1995}~(2006).

\bibitem{digital_zero_noise_extrap}
Tudor Giurgica-Tiron, Yousef Hindy, Ryan LaRose, Andrea Mari, and William~J.
  Zeng.
\newblock ``Digital zero noise extrapolation for quantum error mitigation''.
\newblock In 2020 IEEE International Conference on Quantum Computing and
  Engineering (QCE).
\newblock \href{https://dx.doi.org/10.1109/QCE49297.2020.00045}{Pages
  306--316}.
\newblock ~(2020).

\bibitem{MichaelKrebsbachOptimizingRichardson}
Michael Krebsbach, Bj\"orn Trauzettel, and Alessio Calzona.
\newblock ``Optimization of richardson extrapolation for quantum error
  mitigation''.
\newblock \href{https://dx.doi.org/10.1103/PhysRevA.106.062436}{Phys. Rev. A
  {\bf 106}, 062436}~(2022).

\bibitem{Trefethen_approx}
Lloyd~N. Trefethen.
\newblock ``Approximation theory and approximation practice, extended
  edition''.
\newblock \href{https://dx.doi.org/10.1137/1.9781611975949}{Society for
  Industrial and Applied Mathematics}. Philadelphia, PA~(2019).

\bibitem{Childs_Trotter_Theory_2021}
Andrew~M. Childs, Yuan Su, Minh~C. Tran, Nathan Wiebe, and Shuchen Zhu.
\newblock ``Theory of trotter error with commutator scaling''.
\newblock \href{https://dx.doi.org/10.1103/PhysRevX.11.011020}{Phys. Rev. X
  {\bf 11}, 011020}~(2021).

\bibitem{ImprovedAccuracyforTrotterSim}
Gumaro Rendon, Jacob Watkins, and Nathan Wiebe.
\newblock ``Improved accuracy for trotter simulations using chebyshev
  interpolation''.
\newblock \href{https://dx.doi.org/10.22331/q-2024-02-26-1266}{Quantum Journal
  {\bf 8}, 1266}~(2024).

\bibitem{watson2024exponentially}
James~D. Watson and Jacob Watkins.
\newblock ``Exponentially reduced circuit depths using trotter error
  mitigation''.
\newblock \href{https://dx.doi.org/10.1103/kw39-yxq5}{PRX Quantum {\bf 6},
  030325}~(2025).

\bibitem{Universal_sampling_lower_bounds_QEM}
Ryuji Takagi, Hiroyasu Tajima, and Mile Gu.
\newblock ``Universal sampling lower bounds for quantum error mitigation''.
\newblock \href{https://dx.doi.org/10.1103/PhysRevLett.131.210602}{Phys. Rev.
  Lett. {\bf 131}, 210602}~(2023).

\bibitem{takagi2022fundamental}
Ryuji Takagi, Suguru Endo, Shintaro Minagawa, and Mile Gu.
\newblock ``Fundamental limits of quantum error mitigation''.
\newblock \href{https://dx.doi.org/10.1038/s41534-022-00618-z}{npj Quantum
  Information {\bf 8}, 114}~(2022).

\bibitem{Quek_tighter_bounds_2024}
Yihui Quek, Daniel~Stilck França, Khatri Sumeet, Johannes~Jacob Meyer, and
  Jens Eisert.
\newblock ``Exponentially tighter bounds on limitations of quantum error
  mitigation''.
\newblock \href{https://dx.doi.org/10.1038/s41567-024-02536-7}{Nature Physics
  {\bf 20}, 1648–1658}~(2024).

\bibitem{larose2022mitiq}
Ryan LaRose, Andrea Mari, Sarah Kaiser, Peter~J Karalekas, Andre~A Alves, Piotr
  Czarnik, Mohamed El~Mandouh, Max~H Gordon, Yousef Hindy, Aaron Robertson,
  et~al.
\newblock ``Mitiq: A software package for error mitigation on noisy quantum
  computers''.
\newblock \href{https://dx.doi.org/10.22331/q-2022-08-11-774}{Quantum {\bf 6},
  774}~(2022).

\bibitem{Richardson1911}
Lewis~Fry Richardson.
\newblock ``Ix. the approximate arithmetical solution by finite differences of
  physical problems involving differential equations, with an application to
  the stresses in a masonry dam''.
\newblock
  \href{https://dx.doi.org/https://doi.org/10.1098/rsta.1911.0009}{Philosophical
  Transactions of the Royal Society of London. Series A, containing papers of a
  mathematical or physical character {\bf 210}, 307--357}~(1911).

\bibitem{Richardson1927}
Lewis~Fry Richardson and J~Arthur Gaunt.
\newblock ``Viii. the deferred approach to the limit''.
\newblock
  \href{https://dx.doi.org/https://doi.org/10.1098/rsta.1927.0008}{Philosophical
  Transactions of the Royal Society of London. Series A, containing papers of a
  mathematical or physical character {\bf 226}, 299--361}~(1927).

\bibitem{Extrapolation_book}
Avram Sidi.
\newblock ``Practical extrapolation methods: Theory and applications''.
\newblock \href{https://dx.doi.org/10.1017/CBO9780511546815}{Cambridge
  University Press}. ~(2003).

\bibitem{Atkinson}
Kendall Atkinson.
\newblock ``An introduction to numerical analysis''.
\newblock John wiley \& sons. ~(1991).
\newblock
  url:~\url{https://www.wiley.com/An+Introduction+to+Numerical+Analysis%2C+2nd+Edition-p-9780471624899}.

\bibitem{Runge}
Carl Runge.
\newblock ``Über empirische funktionen und die interpolation zwischen
  \"aquidistanten ordinaten''.
\newblock Zeitschrift f\"ur Mathematik und Physik, vol. 46, pp.
  224–243~(1901).

\bibitem{Townsend}
Laurent Demanet and Alex Townsend.
\newblock ``Stable extrapolation of analytic functions''.
\newblock \href{https://dx.doi.org/10.1007/s10208-018-9384-1}{Foundation of
  Computational Mathematics {\bf 19}, 297–331}~(2016).

\bibitem{demanet2010chebyshev}
Laurent Demanet and Lexing Ying.
\newblock ``On chebyshev interpolation of analytic functions''.
\newblock preprint~(2010).
\newblock
  url:~\url{https://math.mit.edu/sites/icg-new/papers/cheb-interp.pdf}.

\bibitem{cai2021multi}
Zhenyu Cai.
\newblock ``Multi-exponential error extrapolation and combining error
  mitigation techniques for nisq applications''.
\newblock \href{https://dx.doi.org/10.1038/s41534-021-00404-3}{npj Quantum
  Information {\bf 7}, 80}~(2021).

\bibitem{Cheby-poly-book}
J.C. Mason and D.C. Handscomb.
\newblock ``Chebyshev polynomials (1st ed.)''.
\newblock
  \href{https://dx.doi.org/https://doi.org/10.1201/9781420036114}{Chapman and
  Hall/CRC}. ~(2002).

\bibitem{boyd2001chebyshev}
John~P Boyd.
\newblock ``Chebyshev and fourier spectral methods''.
\newblock Courier Corporation. ~(2001).
\newblock  url:~\url{https://link.springer.com/book/9783540514879}.

\bibitem{hairer2006geometric}
Ernst Hairer, Marlis Hochbruck, Arieh Iserles, and Christian Lubich.
\newblock ``Geometric numerical integration''.
\newblock \href{https://dx.doi.org/10.14760/OWR-2006-14}{Oberwolfach Reports
  {\bf 3}, 805--882}~(2006).

\bibitem{watson2024randomly}
James~D Watson.
\newblock ``Randomly compiled quantum simulation with exponentially reduced
  circuit depths''~(2024).
\newblock  \href{http://arxiv.org/abs/2411.04240}{arXiv:2411.04240}.

\bibitem{Qiskit}
Ali Javadi-Abhari, Matthew Treinish, Kevin Krsulich, Christopher~J. Wood, Jake
  Lishman, Julien Gacon, Simon Martiel, Paul~D. Nation, Lev~S. Bishop,
  Andrew~W. Cross, Blake~R. Johnson, and Jay~M. Gambetta.
\newblock ``Quantum computing with qiskit''~(2024).
\newblock  \href{http://arxiv.org/abs/2405.08810}{arXiv:2405.08810}.

\bibitem{hakkaku2025data}
Shigeo Hakkaku, Yasunari Suzuki, Yuuki Tokunaga, and Suguru Endo.
\newblock ``Data-efficient error mitigation for physical and algorithmic errors
  in a hamiltonian simulation''~(2025).
\newblock  \href{http://arxiv.org/abs/2503.05052}{arXiv:2503.05052}.

\bibitem{Hoorfar2008}
Abdolhossein Hoorfar and Mehdi Hassani.
\newblock ``Inequalities on the lambert function and hyperpower function.''.
\newblock \href{https://dx.doi.org/eudml.org/doc/130024}{JIPAM. Journal of
  Inequalities in Pure \& Applied Mathematics [electronic only] {\bf 9}, Paper
  No. 51, 5 p., electronic only--Paper No. 51, 5 p., electronic only}~(2008).

\bibitem{ShortCourseApprox}
Neal~L Carothers.
\newblock ``A short course on approximation theory''.
\newblock Bowling Green State University, Bowling Green, OH{\bf 38}~(1998).
\newblock  url:~\url{https://fourier.math.uoc.gr/~mk/approx1011/carothers.pdf}.

\bibitem{Hoeffding1963}
Wassily Hoeffding.
\newblock ``Probability inequalities for sums of bounded random variables''.
\newblock \href{https://dx.doi.org/10.2307/2282952}{Journal of the American
  statistical association {\bf 58}, 13--30}~(1963).

\end{thebibliography}

\onecolumn
\section{Appendix} \label{appendix} 
\appendix
\renewcommand{\thesubsection}{\Alph{subsection}}

\subsection{Bias Error Analysis}\label{appendix-bia}

In this section, we prove the theorems related to the bias bounds. We use the following lemma in our analysis.
\begin{lemma} \label{LambertW_bound}
    For $0 < \varepsilon < e^{-e}$, if $a \leq e^{-1}$, and $n = \Omega\left(\frac{\log(1/\varepsilon)}{ \sqrt{\log \log (1/\varepsilon)}}\right)$, then $\frac{a^{n}}{n!} \leq \varepsilon $.
\end{lemma}

\begin{proof}
For large enough \(n\), \(n \geq \frac{\log(1/\varepsilon)}{ \sqrt{\log \log (1/\varepsilon)}}\), we let \(y := \log n\), and \(x := \log \frac{1}{\varepsilon}\). Then, \( \log x - \frac{1}{2} \log \log x \leq y\). For \(\varepsilon < e^{-e}\), that is, for \(x > e\), we have the following inequality \cite{Hoorfar2008} for the Lambert-W function  
\begin{equation}
    W_0(x) \leq \log x - \frac{1}{2} \log \log x, 
\end{equation}
implying that \(W_0(x) \leq y\). Since \(W_0(x)\) is increasing on \(x >e\), there exists \(x' \geq x\) such that \(y = W_0(x'\)). By definition of Lambert-W function, \(ye^{y} = x'\). This implies, \(ye^{y} \geq x\), i.e., \(n \log n \geq \log (1/\varepsilon)\), which is equivalent to \(\frac{1}{n^n} \leq \varepsilon\). Since \(a \leq e^{-1}\), using Stirling's approximation, we conclude that \(\frac{a^n}{n!} \leq \varepsilon\).
\end{proof}

\begin{remark}\label{rem-largera}
If the assumption on \(a\) in Lemma 1 does not hold, i.e.,   \(a > e^{-1}\), then we can still show that for \(n = \Omega(\frac{\log {(1/\varepsilon)}}{\sqrt{\log \log (1/\varepsilon) - \log (ae)}})\),  \(\frac{a^n}{n!} \leq \varepsilon\). The proof is the same as Lemma 1 except we set \(x:= \frac{\log (1/\varepsilon)}{ae}\),  \(y:= \log\frac{n}{ae}\), and they satisfy the inequality \( \log x - \frac{1}{2} \log \log x \leq y\). 
\end{remark}

\begin{remark}\label{Rem:largeM}
    The analysis in this Lemma is typically applied to a regime where $\varepsilon \to 0.$ Another important regime is where $\varepsilon$ fixed, while $a\gg 1.$ Again by using the Lambert W function, we find that,
    \[
      n \log \frac{n}{ae} = \log \frac{1}{\varepsilon} \Rightarrow \log \frac{n}{ae} =W_0 \left( \frac{ \log \frac{1}{\varepsilon} }{ae} \right).
    \]
 Again, we fix $\varepsilon<1 $, and consider $a\gg 1$. Using $W_0(x) \leq \log (x+1)$ (from \cite[Theorem 2.3]{Hoorfar2008} by setting $y=1$), we find that 
 \begin{equation}
     \log \frac{n}{ae} = \Omega\left( \frac{ \log \frac{1}{\varepsilon} }{ae}\right)  \Rightarrow n = \Omega \left( a e \varepsilon^{-(1/ae)} \right). 
 \end{equation}
 Consequently, $n$ will mainly scale linearly with $a.$ 
\end{remark}

\textbf{Theorem \ref{thm:equidist_error}. }
\textit{Let $\{x_j\}_{j = 0}^{n}$ be a set of equidistant nodes according to \eqref{xj-unif}, and let \( p_n(x) \) be the degree-\( n \) polynomial interpolating the function \( f \) at these points. Under assumption \eqref{eq:gevrey},  \(\forall\,  0 < \varepsilon < e^{-e} \), if \( M \leq B^{-B/(B-1)} \), and $n = \Omega\left(\frac{\log(1/\varepsilon)}{ \sqrt{\log \log (1/\varepsilon)}}\right)$, then 
\(
|f(0) - p_n(0)| < \varepsilon.
\)
}

\begin{proof}  \label{proof:equidist_error}
 Let \(n_0 = \frac{n}{B-1}\). By \eqref{eq:expectation_error_bound}, \(\abs{f(0) - p_n(0)} \leq C \frac{M^{n+1}}{(n+1)!} \prod_{j=0}^{n} x_j\), and for equidistant nodes,  \(\displaystyle \prod_{j=0}^{n} x_j = \frac{(n_0 + n )!}{n_0! n^n_0}\). Now, the error becomes
 \begin{equation}
\frac{M^{n+1}}{(n+1)!} \prod_{j=0}^{n} x_j =  M^{n+1}
\frac{(n_0 + n )!}{(n+1)!  n_0! n_0^{n}}\,.
\end{equation}

Using Stirling's approximation, $n! \approx \sqrt{2\pi n} (\frac{n}{e})^{n}$, we can simplify the factorial term.

\begin{equation}
  \frac{M^{n+1}}{(n+1)!} \prod_{j=0}^{n} x_j 
  \approx
  e \left(\frac{M}{n+1}\right)^{n+1}   \sqrt{\frac{n_0 + n}{2 \pi n_0 (n+1)}} \left(\frac{n_0 + n}{n_0}\right)^{n_0 + n}\,. 
\end{equation}

Since $\sqrt{\frac{n_0 + n}{2 \pi n_0 (n+1)}} \leq 1$, we have

\begin{equation}
     \frac{M^{n+1}}{(n+1)!} \prod_{j=0}^{n} x_j
     \leq e \left(\frac{M}{n+1}\right)^{n+1} (1 + \frac{n}{n_0})^{n_0 + n}.
\end{equation}

Moreover,  \(n_0 = \frac{n}{B-1}\), then we can simplify the interpolation error to

\begin{equation}
    \frac{M^{n+1}}{(n+1)!} \prod_{j=0}^{n} x_j
    \approx e \left(\frac{M}{n+1}\right)^{n+1} B^{(nB/(B-1))}
     \leq  e \left(\frac{M B^{B/(B-1)}}{n+1}\right)^{n+1}.
\end{equation}

Using Stirling's formula again, we get
\begin{equation} \label{ineq:equally_spaced_error_bound}
    \frac{M^{n+1}}{(n+1)!} \prod_{j=0}^{n} x_j
    \approx    \sqrt{2 \pi (n+1)} \frac{\left({\frac{MB^{B/(B-1)}}{e}}\right)^{n+1}}{(n+1)!}.
\end{equation}

Invoking Lemma \ref{LambertW_bound}, and since $n = \Omega\left(\frac{\log (1/\varepsilon)}{ \sqrt{\log \log (1/\varepsilon)}}\right)$, and \(M \leq B^{-B/(B-1)}\), the proof is complete.

\end{proof}

We proceed by establishing an analogous result for interpolation using Chebyshev nodes.
\medskip

\textbf{Theorem \ref{thm:Cheby_error_theorem}.}
\textit{Let $\{x_j\}_{j = 0}^{n}$ be a set of  Chebyshev nodes in the interval \([1,B], B > 1\), and let \( p_n(x) \) be a degree-$n$ polynomial interpolating the function \( f \) at these points. Under assumption \eqref{eq:gevrey},
for $0 < \varepsilon < e^{-e}$, if \(M \leq 4/((B-1) e \kappa^2)\), and \( n = \Omega \left( \frac{\log (1/\varepsilon)}{ \sqrt{\log \log (1/\varepsilon)}}\right)\), then 
\(
|f(0) - p_n(0)| < \varepsilon.
\)}

\begin{proof} \label{proof:Cheby_error_theorem}
By inequality \eqref{eq:expectation_error_bound}, \(\displaystyle \abs{f(0) - p_n(0)} \leq C \frac{M^{n+1}}{(n+1)!} \prod_{j=0}^{n} x_j\). For Chebyshev nodes, equations \eqref{newT}, and \eqref{roots_Chebyshev_equivalence} imply that

\begin{equation} \label{eq:multi_nodes_chebyshev}
     \displaystyle \abs{\prod_{j=0}^{n} x_j} \leq 
     2 \left(\frac{B-1}{4}\right)^{n+1} \abs{\Tilde{T}_{n+1}(0)}.
\end{equation}

We note that $\tilde{T}_{n+1}(0) = T_{n+1}(- \frac{B+1}{B-1})$.
For every \( x \in \mathbb{R}\) it can be proved that \cite{ShortCourseApprox}
\begin{equation} \label{Chebyshev_formula}
    T_n(x) = \frac{1}{2} [(x+ \sqrt{x^2 - 1})^n + (x- \sqrt{{x^2 - 1}})^n ] 
\end{equation}

In particular, for all $x < 0$, the following inequality is true
\begin{equation}
    T_n(x) \leq \abs{x - \sqrt{x^2 - 1}}^n.
\end{equation}

In our problem, \( x = - 1 - z, \text{where}  \, z = \frac{2}{B-1} \in \mathbb{R}^{+}\). Therefore,
\[
T_n(-1-z) \leq | 1 + z + \sqrt{z^2 + 2 z}|^n 
= |\sqrt{\frac{z}{2}} + \sqrt{1 + \frac{z}{2}}|^{2n},
\]
which implies,
\begin{equation}\label{T0_bound}
\abs{\tilde{T}_{n}(0)} \leq \kappa^{2n}.    
\end{equation}

The last inequality and \eqref{eq:multi_nodes_chebyshev}, along with definition \eqref{def:kappa} result in the following bound,
\begin{equation} \label{ineq:Cheby_error_bound}
    \frac{M^{n+1}}{(n+1)!} \prod_{j=0}^{n} x_j  \leq  \frac{2}{(n+1)!} 
    \left(\frac{M(B-1)\kappa^2}{4}\right)^{n+1}.
\end{equation}

Since \(\frac{M(B-1)}{4} \kappa^2 \leq e^{-1}\), and \( n = \Omega \left( \frac{\log (1/\varepsilon)}{ \sqrt{\log \log (1/\varepsilon)}}\right)\), Lemma \ref{LambertW_bound} implies that \(\frac{M^{n+1}}{(n+1)!} \displaystyle \prod_{j=0}^{n} x_j < \varepsilon\). 
\end{proof}

To find the optimal degree for the polynomial \(p_m\), we continue with the following lemma in order to apply the approximation methods in \cite{Trefethen_approx}.

\begin{lemma} \label{lemma:f_ac_bdd_var}
\textit{
    Let \(f(x)\) be an analytic function on the interval \([1,B]\) with \(B > 1\), under the condition \eqref{eq:gevrey}.  For every \(k \geq 1\), \(f^{(k)}\) is of bounded variation, and \(f\) and its derivatives up to \(f^{(k)}\) are absolutely continuous on \([1,B]\). }
\end{lemma}

\begin{proof} 
For every \(k \geq 1\), and sub-intervals \([a_i,b_i]\) of \([1,B]\), \(\norm{f^{(k+1)}}_1 = \int_{a_i}^{b_i} \abs{f^{(k+1)}(x)}\, dx \leq CM^{(k+1)} (b_i - a_i)\). That is, the \(L_1\)-norm of the derivative of \(f^{(k)}\) is finite, which means \(f^{(k)}\) is a function of bounded variation.

For all \(k \geq 1\), \(f^{(k)}\) is bounded on \([1,B]\). This implies \(f\), and its derivatives up to \(f^{(k)}\) are Lipschitz continuous, which in turn results in them being continuous.
\end{proof}

We can now derive the bound for the degree of \(p_m\), when the bias is less than \(\varepsilon\), under the following conditions.

\textbf{Theorem \ref{thm:bias_bound_lsa}.}
   \textit{Let \(f(x)\) be under the condition \eqref{eq:gevrey}, for \(M < 1\), and let \(p_m(x)\) denote a polynomial of degree \(m \leq n\) obtained through the least squares approximation of \(f\), interpolating the function \( f \) at a set of  Chebyshev nodes in the interval \([1,B], B > 1\). For every \(\varepsilon > 0\),  \(\abs{p_m(0) - f(0)} < \varepsilon\) holds, provided that  \( \kappa \leq M^{-\mu m/2n}\) for some \(\mu \in (0,1)\), and 
\[
    m = \Omega \left( \frac{ \log (C^{'}/\varepsilon)
   }{(1-\mu) \log (1/M)} \right), 
\]
{where} $ C^{'} = \frac{2 (B-1) C M}{ \pi} (\frac{1}{1 -M \kappa^2} + \frac{1}{1 - M}). $ }

\begin{proof} \label{proof:bias_bound_lsa}
Using the equations \eqref{eq:f_cheby_series}, and \eqref{eq:least_sq_approx}, \(f(0) = \displaystyle\sum_{k=0}^{\infty} a_k \tilde{\tau}_k(0)\), and \(p_m(0) =  \displaystyle \sum_{k = 0}^{m} \mathbf{c}_k \tilde{\tau}_k(0)\).
Theorem 4.1, and 4.2 from \cite{Trefethen_approx} prove that 
\begin{equation} \label{eq:error_cheby_series}
    \displaystyle p_m(0) - f(0) 
    = \sum_{k = m + 1}^{\infty} a_k (\tilde{\tau}_k(0) - \tilde{\tau}_t(0)),
\end{equation}
where \(t = \abs{(k+n-1)(\text{mod} \ 2n) - (n-1)}\) is a number in the range \(0 \leq t \leq n \). We note that the above equation also holds in the least squares problem, since \(p_m\) is in essence, an interpolating polynomial of degree \(m\).
 It follows from Lemma \ref{lemma:f_ac_bdd_var}, and Theorem 7.1 in \cite{Trefethen_approx} that for all \(k \geq 0\),
\begin{equation} \label{ak_bound}
    \abs{a_k} \leq \frac{\sqrt{2n+2}}{\pi} (B-1) C M^k,
\end{equation}
where \((B-1)CM^k\) is the bounded variation of \(f^{(k-1)}\), and \(\sqrt{2n+2}\) is due to the normalization term for the Chebyshev polynomials in \eqref{rescaled_Chebypoly}. By \eqref{T0_bound}, and \eqref{rescaled_Chebypoly}, for \(k \geq 1\) we have \(
\abs{\tilde{\tau}_k (0)} \leq \sqrt{\frac{2}{n+1}} \kappa^{2k}
\).   Since \( t \leq n\), \(
\abs{\tilde{\tau}_t (0)} \leq \sqrt{\frac{2}{n+1}} \kappa^{2n}
\). Applying the geometric series results in
\begin{equation}
    \abs{\displaystyle p_m(0) - f(0) }  \leq 
    \frac{2 (B-1) C}{\pi} M^{m+1} \left(\frac{\kappa^{2m+2}}{1 - M\kappa^2} + \frac{\kappa^{2n}}{1-M}\right).
\end{equation}
Since  \( \kappa \leq M^{-\mu m/2n}\), and \(m+1 \leq n\), we can simplify the right hand side further
\begin{equation}
    \abs{\displaystyle p_m(0) - f(0) }  \leq \frac{2 (B-1) C M}{ \pi} M^{(1-\mu)} \left(\frac{1}{1 -M \kappa^2} + \frac{1}{1 - M}\right).
\end{equation}

By the assumption that \(m = \Omega \left( \frac{2 \log (C^{'}/\varepsilon)
   }{(1-\mu)\log (1/M)} \right) \), and \(C^{'} = \frac{2 (B-1) C M}{ \pi} (\frac{1}{1 -M \kappa^2} + \frac{1}{1 - M}) \), we conclude that \(\abs{\displaystyle p_m(0) - f(0) }  < \varepsilon\).
\end{proof}

We also established a bound on the bias arising from the simultaneous mitigation of circuit noise and algorithmic error in the Trotter decomposition, leading to the following theorem:\\

\textbf{Theorem \ref{thm:trotter_error}.}
\textit{ Let \(\rho\) be the density matrix in the  modified Lindblad equation \eqref{eq: modified-lindblad}, and $f(\lambda)= \tr(\rho(t, \lambda) A)$. Also, let \(p_n(x)\) denote the polynomial interpolant of \(f\), where \(f\) and \(p_n\) coincide on Chebyshev nodes \(\{x_i\}_{i=0}^{n}\) in the interval \([1,B]\), \(B > 1\).
     If \(n = \Omega \left(\log \varepsilon/\log (\frac{(B-1) e \kappa^2 \theta}{4 (1 - \lambda \theta)})\right)\), then \(\abs{f(0) - p_n(0)} < \varepsilon\). }

\begin{proof} \label{proof:trotter_error}
By \cref{eq:interp_error} the interpolation error is \(\displaystyle f(0) - p_n(0) = (-1)^{n}f^{(n+1)} (\xi)\frac{\prod_{j=0}^{n}{x_j}}{(n+1)!}\), then
\begin{equation}
    \abs{f(0) - p_n(0)} \leq \frac{\abs{f^{(n+1)} (\xi)}\abs{ \prod_{j=0}^{n}{x_j}}}{(n+1)!}.
\end{equation}
 Lemma \ref{lemma:trotter_derivative_bound} implies that \(\abs{ f^{(n)}(\lambda) }\leq  C\norm{A} \frac{\theta^n}{ (1-\lambda \theta)^n} n^n\), and inequality \eqref{ineq:Cheby_error_bound}  implies that \(\displaystyle |\prod_{j=0}^{n} x_j| \leq 
     2 (\frac{(B-1)\kappa^2}{4})^{n+1} \). Together with Stirling's formula, we conclude that
\begin{equation}
    \abs{f(0) - p_n(0)} \leq \frac{C \norm{A}}{\sqrt{2 \pi (n+1)}} (\frac{(B-1) e \kappa^2 \theta}{4 (1 - \lambda \theta)})^{n+1}.
\end{equation}
Since \(n = \Omega (\log \varepsilon/\log(\frac{(B-1) e \kappa^2 \theta}{4 (1 - \lambda \theta)}))\), \(\abs{f(0) - p_n(0)} < \varepsilon\).
\end{proof}

\subsection{Statistical Error Analysis}

 As mentioned in the main text, to derive an upper bound for the variance of the approximation, it suffices to bound \( \norm{\bm \gamma}_1 \).
 \medskip

\textbf{Theorem \ref{equidist_gamma_sum_bound}.}
\textit{Let $\{x_j\}_{j = 0}^{n}$ be a set of equidistant nodes according to \eqref{xj-unif}. Then, the following holds:
    \[
    \norm{\bm \gamma}_1 = O\left( \left( \frac{2Be}{B-1} \right)^n \right).
    \]}
    
\begin{proof} \label{proof:equidist_gamma_sum_bound}

 By definition \eqref{gamma_j}, we have,
\(
     \gamma_j = \displaystyle \prod_{\substack{k=0 \\ k \neq j}}^{n} \frac{h +k}{k-j}. 
\)
 Simplifying it results in 
\begin{equation}
    \gamma_j = \frac{(-1)^{j} n!}{(h +j) j! (n-j)!} \frac{\prod_{k=0}^{n} (h + k)}{n!}.
\end{equation}
Let $c_{n,h} := \frac{\prod_{k=0}^{n} (h + k)}{n!}$, then we have,
\(
    \gamma_j = \frac{(-1)^{j}}{h + j} c_{n,h} \binom{n}{j}
\), 
yielding the bound 
\begin{equation} \label{equidist_gamma_sum}  
   \left\|\bm \gamma \right\|_1 = c_{n,h} \sum_{j=0}^{n} \frac{\binom{n}{j}}{h + j} \leq \frac{c_{n,h} \sum_{j=0}^{n} \binom{n}{j}}{h} = \frac{2^n c_{n,h}}{h}.
\end{equation}
Using Stirling's formula, and considering the fact that \(h = \frac{n}{B-1}\),  we get
\begin{equation} \label{c_n,h bound}
    c_{n,h} \leq \frac{(h + n)^{n+1}}{n!} \leq \frac{ h^{n+1} B^{n+1}}{n!}.
\end{equation}
Consequently, we obtain,
\begin{equation}
    \norm{\bm \gamma}_1 \leq 
    B  \left(\frac{2Be}{B-1}\right)^n
\end{equation}

\end{proof}

To evaluate the effect of the distribution of nodes on variance, we also compute \(\norm{\bm \gamma}_1\) for the Chebyshev points, leading to the following theorem.
\medskip

\textbf{Theorem \ref{thm:Cheby_gamma_sum_bound}.}
   \textit{Let $\{x_j\}_{j = 0}^{n}$ be the Chebyshev nodes in the interval \([1,B]\), where \( B > 1 \). Then, we have the bound:
\[
    \norm{\bm \gamma}_1 = O\left( \kappa^{2n}  \right),
\]
where $\kappa$ is defined in \eqref{def:kappa}.
}

\begin{proof} \label{proof:Cheby_gamma_sum_bound}
     Let $\Phi_{j, n+1} (x) := \displaystyle \prod_{j=0}^{n} (x-x_j) $, where $x \in \mathbb{R}$. The \(j\)-th Lagrange basis polynomial can be written as follows:

\begin{equation} \label{lagrange_polynomial}
L_j(x) = 
\begin{cases}
    1, & \text{if } x = x_j, \\
    \frac{\Phi_{j, n+1} (x)}{(x - x_j) \Phi_{j, n+1}'(x_j)}, & \text{if } x \neq x_j.
\end{cases}
\end{equation}
For $x_j \in [1,B]$, since \(\gamma_j = L_j(0)\), \cref{roots_Chebyshev_equivalence} implies that
\begin{equation}
    \gamma_j
    =
    \frac{\tilde{T}_{n+1}(0)}
    {-x_j\tilde{T}^{'}_{n+1}(0)}.
\end{equation}
From  \eqref{newT} we have
\(\tilde{T}_{n}(0) = T_{n}(- \frac{B+1}{B-1})\). Also, for \(  y\in [-1,1]\),  \(
  T_{n+1}(y) = \cos( (n+1) \arccos y)\), then a direct calculation yields that
\begin{equation}
     T_{n+1}'(y_j) = - \frac{n+1}{\sqrt{1 - y_j^2}} \sin ( (n+1)  \arccos y_j ) 
 \Rightarrow
 \abs{T_{n+1}'(y_j) } = \frac{n+1}{ \abs{\sin (\frac{2j+1} {n+1}\frac{\pi}{2}) } } \geq n+1.
\end{equation} 
Together with the observation that $\abs{x_j} \geq 1$ (\(x_j = \phi (y_j) \) as in \eqref{Chebyshev_shift}), and equation \eqref{T0_bound} we have,

\begin{equation} \label{gamma-bound-cheby}
 \abs{\gamma_j} \leq \frac{1}{n+1} (\frac{\sqrt{B}+ 1}{\sqrt{B}-1})^{2(n+1)}  \Rightarrow \sum_{j=0}^n \abs{\gamma_j} \leq  \kappa^{2n+2}.    
\end{equation}
\end{proof}

\textbf{Theorem \ref{thm:LSA_gamma_sum}.}
\textit{Let $\{x_j\}_{j = 0}^{n}$ be Chebyshev nodes in the interval \([1, B]\), \( B > 1 \),  \(\{\tilde{\tau}_k(x)\}_{k=0}^{m}\) be shifted normalized Chebyshev polynomials of the first kind on \([1,B]\), and \(  \displaystyle \norm{\bm \gamma}_1=\sum_{i = 0}^{n}  \abs{\gamma_i}
\), where \( \gamma_i :=  \displaystyle \sum_{k=0}^{m} \tilde{\tau}_k(x_i) \tilde{\tau}_k(0)\). Then,
    \begin{equation}
        \norm{\bm \gamma}_1= O\left(\kappa^{2m}\right). 
    \end{equation}
}

\begin{proof} \label{proof:LSA_gamma_sum}
 Using \eqref{rescaled_Chebypoly}, \(\abs{\tilde{\tau}_k(x_i)} \leq \sqrt{\frac{2}{n+1}} \), and by \eqref{T0_bound},   
\(
    \abs{\tilde{\tau}_k(0)} \leq \sqrt{\frac{2}{n+1}} \kappa^{2k}
\). Thus, \cref{tau2gamma} implies that
\begin{equation}
   \norm{\bm \gamma}_1 = \sum_{i = 0}^{n} \sum_{k=0}^{m} \abs{\tilde{\tau}_k(x_i)} \abs{\tilde{\tau}_k(0)} = \sum_{k=0}^{m} \sum_{i = 0}^{n} \abs{\tilde{\tau}_k(x_i)} \abs{\tilde{\tau}_k(0)}. 
\end{equation}
Applying Cauchy-Schwarz inequality leads to \(\displaystyle \sum_{i = 0}^{n} \abs{\tilde{\tau}_k(x_i)} \leq (\sum_{i=0}^{n} \tilde{\tau}^2_k(x_i))^{1/2} (\sum_{i=0}^{n} 1)^{1/2}\). It follows that,
\begin{equation}
    \norm{\bm \gamma}_1 =  \sqrt{2} \sum_{k=0}^{m} \kappa^{2k} = \sqrt{2} \frac{\kappa^{2m+2} - 1}{\kappa^2 - 1}.
\end{equation}
Consequently, \(\norm{\bm \gamma}_1 = O(\kappa^{2m}).\)
\end{proof}

\subsection{Bounding Sampling Complexity}
To bound the sampling complexity in terms of bias and statistical error, we begin by stating Hoeffding's inequality for bounded random variables \cite{Hoeffding1963}. 

Let \(X_1, \dots, X_{N_S}\) be independent random variables satisfying \( a_i \leq X_i \leq b_i \) almost surely. Consider their sum,
\[
S_{N_S} = X_1 + \dots + X_{N_S}.
\]
Hoeffding's inequality Under a\cite{Hoeffding1963} states that for every \( t > 0 \),
\begin{equation} \label{Hoeff}
 \mathbb{P}\left(|S_{N_S} - \mathbb{E}[S_{N_S}]| \geq t\right) \leq 2 \exp\left(- \frac{2t^2}{\sum_{i=1}^{N_S} (b_i - a_i)^2}\right).
\end{equation}
\medskip
\textbf{Theorem \ref{thm:Hoeffding_Richardson}.} 
\textit{Let \(f\), and \(p_{n}(x)\) be as in \cref{eq:f-definition} and \cref{def:extrap-pn} respectively.  Under assumption \eqref{eq:gevrey}, \( \forall\, 0 < \varepsilon < e^{-e}, 0 < \delta < 1\), the error bound \(\abs{p_n(0) - f(0)} < \varepsilon\) holds with probability at least \(1 - \delta\)  if \(n = \Omega \left( \log \frac{1}{\varepsilon} \right)\), and and the number of samples \( N_S \) satisfies either one of the following conditions:}
\begin{enumerate}
    \item \textit{$\{x_j\}_{j = 0}^{n}$ is a set of equidistant nodes according to \eqref{xj-unif}, \( M \leq B^{-B/(B-1)} \), and
    \[
        N_S = \Omega\left(\alpha^2  \varepsilon^{-(2+ 4Be/(B-1))}\log \frac{2}{\delta}\right).
    \]}

    \item \textit{$\{x_j\}_{j = 0}^{n}$ is a set of Chebyshev nodes in the interval \([1,B] \), \(M \leq 4/((B-1) e \kappa^2)\),  and 
    \[
        N_S = \Omega\left(\alpha^2  \varepsilon^{-(2 + 4 \log \kappa)}\log \frac{2}{\delta}\right),
    \]}
\end{enumerate}
where \(\kappa\) is defined in \eqref{def:kappa}.

\begin{proof} \label{proof:Hoeffding_Richardson}
Let the random variable associated with each measurement be \( X_j^i \) (\( i = 1, \ldots, N_S \)). Since \(\|A\|_\infty \leq \alpha\),
\begin{equation}
   \displaystyle \abs{X_i} \leq \frac{\alpha \displaystyle\sum_{j=0}^{n} \abs{\gamma_j}}{N_S}.  
\end{equation}
Define \(S_{N_S} := \displaystyle \sum_{i=1}^{N_S} X_i\), then, \(\mathbb{E}(S_{N_S}) = p_n(0)\), and \(\abs{S_{N_S} } \leq \alpha \displaystyle \sum_{j=0}^{n} \abs{\gamma_j} \). Hoeffding inequality \eqref{Hoeff} implies,
\begin{equation}
\displaystyle \mathbb{P}(|S_{N_S} - \mathbb{E}(S_{N_S})| \geq \varepsilon) \leq 
2 \exp(- \frac{\varepsilon^2 N_S}{2 \alpha^2 (\displaystyle \sum_{j=0}^n  \abs{\gamma_j} )^2}).    
\end{equation}
We want \(\delta > 0\), such that 
 \(
     \mathbb{P}(|S_{N_S} - \mathbb{E}(S_{N_S})| \geq \varepsilon) \leq \delta
 \), which is possible when 
\(
 2 \exp(- \frac{\varepsilon^2 N_S}{2 \alpha^2 (\displaystyle \sum_{j=0}^n  \abs{\gamma_j} )^2})   \leq \delta.
\)
As defined earlier, \(\displaystyle \norm{\bm \gamma}_1 = \sum_{j=0}^{n} \abs{\gamma_j}\). It changes according to the distribution of nodes:
\begin{enumerate}
 \item Equidistant nodes: For \( 0 < \varepsilon < e^{-e} \), 
\[
\log \left( \frac{1}{\varepsilon} \right) \geq \frac{\log \left( 1/\varepsilon \right)}{\sqrt{\log \log \left( 1/\varepsilon \right)}} ,
\]
from which it follows \( n = \Omega \left( \frac{\log \left( 1/\varepsilon \right)}{\sqrt{\log \log \left( 1/\varepsilon \right)}} \right) \).
   Thus, 
\(|p_n(0) - f(0)| < \varepsilon \) 
   by Theorem \ref{thm:equidist_error},  since \( M \leq B^{-B/(B-1)} \), and $n = \Omega\left(\frac{\log(1/\varepsilon)}{ \sqrt{\log \log (1/\varepsilon)}}\right)$.
   Following Theorem \ref{equidist_gamma_sum_bound},
    \( \norm{\bm \gamma}_1 = O\big((\frac{2Be}{B-1})^n\big)\), then

\begin{equation}
 2 \exp(- \frac{\varepsilon^2 N_S}{2 \alpha^2  \norm{\bm \gamma}^2_1})   \leq  2 \exp (- \frac{\varepsilon^2 N_S}{2 \alpha^2 (\frac{2Be}{B-1})^{2n}}) \leq \delta.   
\end{equation}

From the right inequality, we get
\begin{equation}
    N_S \geq \frac{2 \alpha^2 (\frac{2Be}{B-1})^{2n} \log \frac{2}{\delta}}{\varepsilon^2}. 
\end{equation}
 Since \(n = \Omega(\log (\frac{1}{\varepsilon}))\), 
\(
    N_S \geq  (2 \alpha^2 \log \frac{2}{\delta}) (\frac{2Be}{B-1})^{2\log (1/\varepsilon)} /\varepsilon^2
\). By simplifying we get,
\begin{equation}
    N_S \geq (2 \alpha^2 \log \frac{2}{\delta}) \varepsilon^{-2- (4Be)/(B-1)}.
\end{equation}

\item  Chebyshev nodes: By Theorem \ref{thm:Cheby_error_theorem}, for \(0 < \varepsilon < e^{-e}\), \(|p_n(0) - f(0)| < \varepsilon\), since  \(M \leq 4/((B-1) e \kappa^2)\), and \( n = \Omega ( \frac{\log (1/\varepsilon)}{\sqrt{\log \log (1/\varepsilon)}} )\). By Theorem \ref{thm:Cheby_gamma_sum_bound}, \(\displaystyle \norm{\bm \gamma}_1 = O\left( \kappa^{2n}  \right)\), then
\begin{equation}
 2 \exp(- \frac{\varepsilon^2 N_S}{2 \alpha^2 \norm{\bm \gamma}^2_1})   \leq  2 \exp (- \frac{\varepsilon^2 N_S }{2\alpha^2 \kappa^{4n}}) \leq \delta  . 
\end{equation}
It follows that
\(
    N_S \geq (2\alpha^2  \kappa^{4n}\log \frac{2}{\delta})/\varepsilon^2 .
\) This yields the bound,
\begin{equation}
N_S \geq (2\alpha^2  \log \frac{2}{\delta} )\varepsilon^{-2 - 4 \log \kappa}.
\end{equation}

\end{enumerate}

In both cases 
\(|p_n(0) - f(0)| \leq |S_{N_S} - \E(S_{N_S})|\). Therefore, \(|p_n(0) - f(0)| < \varepsilon \), with probability \(1- \delta\).
\end{proof}
Similarly, we determine the required number of samples when we extrapolate using least squares as follows.

\medskip
\textbf{Theorem \ref{thm:Hoeffding_lsa_thm}.} 
Let \(f\), and \(p_{n}(x)\) be as in \cref{eq:f-definition} and \cref{eq:least_sq_approx} respectively.  Under assumption \eqref{eq:gevrey}, \( \forall\, 0 < \varepsilon , 0 < \delta < 1\), the error bound \(\abs{p_n(0) - f(0)} < \varepsilon\) holds with probability at least \(1 - \delta\)  if  \( \kappa \leq M^{-\mu m/2n}\) for \(\mu \in (0,1)\), \(m = \Omega \left( \frac{ \log (C^{'}/\varepsilon)
   }{(1-\mu) \log (1/M)} \right)\), and the number of samples \( N_S \) satisfies the following condition:
\[
N_S = \Omega \left((2 \alpha^2 \kappa^{\log C'/(1-\mu)} \log \frac{2}{\delta}) \varepsilon^{-2 -\log \kappa/(1-\mu)}\right),
\]
 where \(C^{'} = \frac{2 (B-1) C M}{ \pi} (\frac{1}{1 -M \kappa^2} + \frac{1}{1 - M}) \).

\begin{proof} \label{proof:Hoeffding_lsa_thm}
    The idea is similar to the proof of Theorem \ref{thm:Hoeffding_Richardson}.  If  \( \kappa \leq M^{-\mu m/2n}\) for \(\mu \in (0,1)\), and \(m = \Omega \left( \frac{ \log (C^{'}/\varepsilon)
   }{(1-\mu) \log (1/M)} \right)\), by Theorem \ref{thm:bias_bound_lsa}, for every \(\varepsilon > 0\),  \(\abs{p_m(0) - f(0)} < \varepsilon\).  Theorem \ref{thm:LSA_gamma_sum} implies that \(\displaystyle\norm{\bm \gamma}_1 = O\left(\kappa^{2m}\right)\).
   Let the random variable associated with each measurement be \( X_j^i \) (with \( i = 1, \ldots, N_S \)). We continue as in the proof of Theorem \ref{thm:Hoeffding_Richardson}, and define \(S_{N_S} := \displaystyle \sum_{i=1}^{N_S} X_i\). Thus, \(\E(S_{N_S}) = p_n(0)\), and \(\abs{S_{N_S} } \leq \alpha \displaystyle \sum_{j=0}^{n} \abs{\gamma_j} \). Using Hoeffding's inequality \eqref{Hoeff}, we want to find \(N_S\), such that for \(\delta > 0\),

\begin{equation}
\displaystyle \mathbb{P}(|S_{N_S} - \mathbb{E}(S_{N_S})| \geq \varepsilon) \leq 
 2 \exp(- \frac{\varepsilon^2 N_S}{2 \alpha^2  \norm{\bm \gamma}^2_1})    \leq \delta .
\end{equation}
It follows that,

\begin{equation}
 2 \exp(- \frac{\varepsilon^2 N_S}{2 \alpha^2  \norm{\bm \gamma}^2_1})    \leq  2 \exp (- \frac{\varepsilon^2 N_S}{2\alpha^2 \kappa^{2m}}) \leq \delta .  
\end{equation}
From the right inequality, we have,
\begin{equation}
    N_S \geq \frac{2\alpha^2 \kappa^{2m} \log \frac{2}{\delta}}{\varepsilon^2} .
\end{equation}
Since \(m = \Omega \left( \frac{ \log (C^{'}/\varepsilon)
   }{(1-\mu) \log (1/M)} \right)\), the above inequality can be written as:

\begin{equation}
     N_S \geq 2 \alpha^2 \kappa^{\log C'/(1-\mu)} \log \frac{2}{\delta} \varepsilon^{-2 -\log \kappa/(1-\mu)}.
\end{equation}
\end{proof}

\subsection{Derivative Bounds for Expectation Values} \label{app:exp_value_bound}

Following the Lindblad equation \eqref{eq:lindblad_eq}, we write the equation for the density operator as
\begin{equation}  \label{Lindblad_gamma}
    \frac{\partial}{\partial t} \rho (t) =  -i [H(t), \rho(t)] + \lambda \mathcal{L}(\rho(t)),   
\end{equation}
where \(\mathcal{L}\) is a linear operator, and \( \lambda := \lambda_0 x\) is the noise strength, for \(x \in [1,\infty)\). Let \( \Gamma_1 (t) := \frac{\partial \rho (t)}{\partial x}\) be the first-order derivative of \(\rho\) with respect to \(x\), then, 
\begin{equation}  
    \frac{d}{d t} \Gamma_1 (t) =  -i [H(t), \Gamma_1(t)] + \lambda \mathcal{L}(\Gamma_1(t)) + \lambda_0  \mathcal{L}(\rho(t)).   
\end{equation}
By definition
\begin{equation}
    \Gamma_1 (0) = 0.
\end{equation}
Using the generator of Markovian dynamics \(\mathcal{L}_I\) of \cref{eq:lindblad_eq}, we can rewrite \eqref{Lindblad_gamma} as 

\begin{equation} \label{eq:Lindblad_gamma1}
     \frac{d}{d t} \Gamma_1 (t) = \mathcal{L}_I \Gamma_1 (t) + \lambda_0 \mathcal{L}(\rho (t)).
\end{equation}
Then by Duhamel's principle, one has 

\begin{equation}
    \Gamma_1 (t) = \int_0^t e^{(t-t_0)\mathcal{L}_I} (\lambda_0 \mathcal{L}(\rho (t_0)) \, d t_0 .
\end{equation}

More generally, let \( \Gamma_k (t) := \frac{\partial^k \rho (t)}{\partial x^k}\) denote the \(k-\)th derivative of \(\rho\) with respect to \(x\). Then, by induction for a positive integer \(k\), we have,
\begin{equation} 
     \frac{d}{d t} \Gamma_k (t) = \mathcal{L}_I (\Gamma_k (t)) + k \lambda_0 \mathcal{L}(\Gamma_{k-1} (t)),
\end{equation}
where \(\Gamma_0(t) := \rho(t,\lambda)\). Since the right-hand side of the equation above still depends on \(\Gamma_{k-1}\), we can apply Duhamel's principle recursively to obtain the following

\begin{equation}
    \Gamma_k(t) := \int_0^t \int_0^{t_{k-1}} \dots \int_0^{t_1} \lambda_0^k k! e^{(t-t_{k-1}) \mathcal{L_I}} \mathcal{L}
    \dots 
    e^{(s_{1}-t_{0}) \mathcal{L_I}} \mathcal{L}\rho_{\lambda} (t_0)
    \, dt_0 dt_1 \dots dt_{k-1}.
\end{equation}
Now we can bound \(\Gamma_k\) directly as
\begin{equation}
    \|\Gamma_k\| \leq \int_{0 \leq t_1 \leq \dots \leq t_{k-1} \leq t} \lambda_0^k k! \|\mathcal{L}\|^k 
    \, dt_0 dt_1 \dots dt_{k-1}.
\end{equation}
Thus, we arrive at a specific bound,  
\begin{equation}
    \|\Gamma_k\| \leq \lambda_0^k \|\mathcal{L}\|^k t^k .
\end{equation}
By definition, \(f^{(k)}(\lambda_0 x) :=  \frac{\partial^k}{\partial x^k} f (\lambda_0 x) = \tr (\Gamma_k (t) A)\), then 
\begin{equation}
    |\tr (\Gamma_k (t) A)| \leq \|\Gamma_k\| \|A\|.
\end{equation}
Then, if the density matrix evolves until time \(T\), we have,

\begin{equation} \label{interpolation_bound}
    |f^{(k)}(\lambda_0 x)| \leq (\lambda_0 \|\mathcal{L}\| T)^k \|A\|.
\end{equation}
This verifies our assumption in \eqref{eq:gevrey}. In particular, we can set $C=\norm{A}$ and $M=\lambda_0 \|\mathcal{L}\| T.$
A similar bound was derived in \cite{KTemme_SBravyi_JGambetta}. We will extend this derivation to incorporate Trotter in the next section.   

\subsection{Derivative Bounds for Trotterized Expectation Values}\label{app:trotter}

We defined the $k$-th derivative of the density matrix with respect to the noise strength \(\lambda\) in  \eqref{Lamk}. By induction, this derivative can be expressed in terms of lower-order derivatives.

\begin{lemma} \label{lemma:trotter_derivative}
    The derivatives \eqref{Lamk} of the density operator satisfy the following recursion relation,
    \begin{equation}\label{recur-rel}
        \frac{d}{dt} \Lambda_k (t) = \mathcal{L} (\Lambda_k (t)) + \sum_{m=1}^k \genfrac{(}{)}{0pt}{}{k}{m} \sum_{j=m}^{\infty} j (j-1) \cdots (j-m+1)\lambda^{j-m} \mathcal{L}_j (\Lambda_{k-m} (t)).
    \end{equation}
\end{lemma}

\begin{proof}
    We follow a standard recursive argument: by \eqref{eq: modified-lindblad} for \(k = 1\),
    \begin{equation} 
        \frac{d}{dt} \Lambda_1(t) = -i [ H(t) , \Lambda_1(t)] + \mathcal{L}(\rho(t)) - i \sum_{l \geq 2} l \lambda^{l-1} [H_{2l}(t) , \rho(t)] + \lambda \mathcal{L}(\Lambda_1 (t)) - i \sum_{l \geq 1} \lambda^l [H_{2l} (t) , \Lambda_1(t)].
    \end{equation}
    Using the notation in \eqref{eq:L-expand}
    \begin{equation}
        \frac{d}{dt} \Lambda_1(t) = \mathcal{L}(\Lambda_1(t)) + \sum_{j=1}^{\infty} j \lambda^{j-1} \mathcal{L}_j (\rho(t)).
    \end{equation}
    We assume \eqref{recur-rel} holds for \(k\), then for \(k+1\),
 \begin{align*}
    &\frac{d}{dt} \Lambda_{k+1}(t) \\&= \mathcal{L} (\Lambda_{k+1}(t)) + 
    \sum_{m=1}^{k} \genfrac{(}{)}{0pt}{}{k}{m} \sum_{j=m}^{\infty} j (j-1) \dots (j-m) \lambda^{j-m-1} \mathcal{L}_j(\Lambda_{k-m}(t)) \\
    &\quad + \sum_{m=1}^{k} \genfrac{(}{)}{0pt}{}{k}{m} \sum_{j=m}^{\infty} j(j-1) \dots (j-m+1) \lambda^{j-m} \mathcal{L}_j(\Lambda_{k-m+1}(t)) \\
    &= \mathcal{L} (\Lambda_{k+1}(t)) + \sum_{m=1}^{k} \genfrac{(}{)}{0pt}{}{k}{m} \sum_{j = m}^{\infty} j \dots (j-m+1) \Big[ (j-m) \lambda^{j-m-1} \mathcal{L}_j (\Lambda_{k-m}(t)) + \lambda^{j-m} \mathcal{L}_j (\Lambda_{k-m+1}(t)) \Big].
\end{align*}
Since \(\genfrac{(}{)}{0pt}{}{k+1}{m} = \genfrac{(}{)}{0pt}{}{k}{m-1} + \genfrac{(}{)}{0pt}{}{k}{m}
\), a shift of indices  in the first term on the right-hand side results in
\begin{equation}
    = \mathcal{L} (\Lambda_{k+1} (t)) + \sum_{m=1}^{k+1} \genfrac{(}{)}{0pt}{}{k+1}{m} \sum_{j=m}^{\infty} j (j-1) \cdots (j-m+1)\lambda^{j-m} \mathcal{L}_j (\Lambda_{k+1-m} (t)).
\end{equation}
This completes the proof.
\end{proof}

The higher-order terms \(\mathcal{L}\) in \cref{recur-rel} correspond to the commutators arising from Trotter decomposition. From \eqref{recur-rel}, we can see that the inner sum can be related to the derivative of the function
\begin{equation}
    S(y)= \sum_{j=2}^{+\infty} y^j= \frac{y^2}{1-y}, 
\end{equation}
where $y$ is assumed to be in $(0,1)$ and it combines $\lambda $ and $\theta$, the latter of which is defined in \cref{Lj-bound}. 

\begin{lemma} \label{lemma:geometric_sum_derivative}
    There exists a constant $C$, such that,
    \begin{equation}
        \abs{S^{(k)}} \leq \frac{C}{(1-y)^k}.
    \end{equation}
\end{lemma}

\begin{proof}
We start with
\[S^{(k)}(y) = \sum_{m=0}^k \binom{k}{m} f^{(m)}(y) g^{(k-m)}(y),\]
 where \( f(y) = y^2 \), and \( g(y) = \frac{1}{1-y} \). The derivatives of \(f(y)\), and \(g(y)\) are given by:
        \[
        f^{(m)}(y) =
        \begin{cases}
            2y & \text{if } m = 1, \\
            2 & \text{if } m = 2, \\
            0 & \text{if } m > 2.
        \end{cases}
        \quad \text{and}\quad
        g^{(n)}(y) = \frac{n!}{(1-y)^{n+1}}.
        \]
Thus, we obtain
\begin{equation}
    S^{'}(y) = \frac{1}{(1-y)^{2}} \left[ y^2 + 2y(1-y)\right], 
\end{equation}
and, for \(k \geq 2\),
\begin{equation}
    S^{(k)}(y) = \frac{k!}{(1-y)^{k+1}} \left[ y^2 + \frac{2y(1-y)}{k} + \frac{2(1-y)^2}{k(k-1)} \right].
\end{equation}
Thus, the lemma holds.
\end{proof}
Using the above estimates, we can directly derive the following bounds
\[
\begin{aligned}
    \norm{\Lambda_1(t)} &=O(\frac{t\theta }{1-\lambda \theta}), \\
 \norm{\Lambda_2(t)} &=O\left(\frac{\theta^2}{ (1-\lambda \theta)^2} (t^2+t) \right), \\
  \norm{\Lambda_3(t)} &=O\left(\frac{\theta^3}{ (1-\lambda \theta)^3} (t^3+3t^2 + 3t) \right).
\end{aligned}
\]
Therefore, we observe that $\norm{\Lambda_k(t)}$ scales polynomially with $t$,
\begin{equation} \label{lambda_k_bound}
\norm{\Lambda_k(t)} =O\left(\frac{\theta^k}{ (1-\lambda \theta)^k} (a_{k,k} t^k+a_{k,k-1} t^{k-1} + \cdots + a_{k,1} t) \right).
\end{equation}

\begin{lemma} \label{lemma:trotter_derivatives_coefficients}
        The coefficients of the polynomial in \eqref{lambda_k_bound} satisfy \( a_{1,1} = 1 \),  \( a_{k,j} = 0, \forall j > k \), and the following recursion relation,
    \begin{equation} \label{a_k_j+1_bound}
         a_{k,j+1} \leq  \sum_{m=1}^{k-j} \binom{k}{m} \frac{1}{j+1} a_{k-m,j}.
    \end{equation}
    Moreover, the coefficients satisfy,
    \begin{equation}
        a_{k,j} \leq \frac{j^k}{j!}, \quad \sum_{j=1}^k a_{k,j} \leq k^k.
    \end{equation}
\end{lemma}

\begin{proof}
    We already know that  \( a_{1,1} = 1 \), and  \( a_{k,j} = 0 \) for \(j > k \). We obtain \eqref{a_k_j+1_bound} by induction on the bounds of $ \norm{\Lambda_k(t)}$. Again by induction,  \( a_{1,1} \leq 1 \). Assume that for all \(k' < k\), and \(j \geq 1\), \( a_{k',j} \leq \frac{j^{k'}}{j!}\). From the recursion relation:
\[
a_{k,j+1} \leq \sum_{m=1}^{k-j} \binom{k}{m} \frac{1}{j+1} a_{k-m,j}.
\]
By the induction hypothesis, \( a_{k-m,j} \leq \frac{j^{k-m}}{j!} \). Substituting this into the recursion:
\[
a_{k,j+1} \leq \frac{1}{(j+1)!} \sum_{m=1}^{k-j} \binom{k}{m} j^{k-m} \leq \frac{1}{(j+1)!}\sum_{m=0}^k \binom{k}{m} j^{k-m} = \frac{(j+1)^k}{(j+1)!}.
\]
We conclude \(a_{k,j+1} \leq \frac{j^k}{j!}\).
To prove \( \displaystyle \sum_{j=1}^k a_{k,j} \leq k^k \), we use the bound \( a_{k,j} \leq \frac{j^k}{j!} \), and consider the following sum:
\[
\sum_{j=1}^k a_{k,j} \leq \sum_{j=1}^k \frac{j^k}{j!}.
\]
\end{proof}

\begin{lemma} \label{lemma:trotter_derivative_bound}
    Let \(\rho\) be the density matrix satisfying the modified Lindblad equation \eqref{eq: modified-lindblad}, where the operator $\mathcal{L} $ on the right-hand side, has the expansion in \eqref{eq:L-expand}, with each operator bounded as in \eqref{Lj-bound}. Then, the expectation of an observable, given by $f(\lambda) = \tr(\rho(t, \lambda) A),$ has derivatives of all orders. When \( t = 1 \), these derivatives satisfy the bound  
    \begin{equation}
       \abs{ f^{(k)}(\lambda) }\leq  C\norm{A} \frac{\theta^k}{ (1-\lambda \theta)^k} k^k.
    \end{equation}
\end{lemma}

\begin{proof}
    By definition \(f^{(k)}(\lambda) = \tr(\Lambda_k(t) A)\), then using \eqref{lambda_k_bound}, and Lemma \ref{lemma:trotter_derivatives_coefficients}, we find that for some constant \(C > 0\),
    \begin{equation}\label{gevrey-1}
        \abs{ f^{(k)}(\lambda) }\leq  C\norm{A}\frac{\theta^k}{ (1-\lambda \theta)^k} k^k.
    \end{equation}
\end{proof}

\vspace{.2in}

\end{document}